\documentclass[sn-mathphys,Numbered]{sn-jnl}

\usepackage{graphicx}%
\usepackage{multirow}%
\usepackage{amsmath,amssymb,amsfonts}%
\usepackage{amsthm}%
\usepackage{mathrsfs}%
\usepackage[title]{appendix}%
\usepackage{xcolor}%
\usepackage{textcomp}%
\usepackage{manyfoot}%
\usepackage{booktabs}%
\usepackage{algorithm}%
\usepackage{algorithmicx}%
\usepackage{algpseudocode}%
\usepackage{listings}%
\usepackage{tikz}%
\usepackage[T1]{fontenc}
\usepackage{graphicx}
\usepackage{amsmath}
\usepackage{amssymb}
\usepackage{tikz}
\usepackage{blkarray}

\usepackage{bm}
\usepackage{algpseudocode}

\usepackage{color} 
\usepackage{amsfonts}
\usepackage{footnote}
\usepackage{multirow,array}

\usepackage{diagbox}
\usepackage{floatflt}
\usepackage{caption}
\usepackage{algorithm}
\usepackage{graphicx}
\usepackage{subcaption}
\usepackage{algorithmicx}
\usepackage{marvosym}
\usepackage{booktabs}
\usepackage{threeparttable}
\usepackage{arydshln}

\usepackage{tabularx}

\theoremstyle{thmstyleone}%
\newtheorem{theorem}{Theorem}
\theoremstyle{thmstyletwo}%
\newtheorem{example}{Example}%
\newtheorem{remark}{Remark}%

\theoremstyle{thmstylethree}%
\newtheorem{definition}{Definition}%

\theoremstyle{thmstylefour}%
\newtheorem{lemma}{Lemma}%

\theoremstyle{thmstylefive}%
\newtheorem{corollary}{Corollary}%

\begin{document}

\title[Article Title]{SXVCS: An XOR-based Visual Cryptography Scheme without Noise via Linear Algebra}


\author*[1]{\fnm{Zizhuo} \sur{Wang}}\email{wangzzh2020@lzu.edu.cn}

\author[1]{\fnm{Ziyang} \sur{Xu}}\email{xuzy20@lzu.edu.cn}

\author[1]{\fnm{Xingxing} \sur{Jia}}\email{jiaxx@lzu.edu.cn}

\affil*[1]{\orgdiv{School of Mathematics and Statistics}, \orgname{Lanzhou University}, \orgaddress{\street{222 South Tianshui Road}, \city{Lanzhou}, \postcode{730000}, \state{Gansu}, \country{P.R. China}}}




\abstract{Visual Cryptography Schemes (VCS) based on the "XOR" operation (XVCS) exhibit significantly smaller pixel expansion and higher contrast compared to those based on the "OR" operation. Moreover, the "XOR" operation appears to possess superior qualities, as it effectively operates within a binary field, while the "OR" operation merely functions as a ring with identity. Despite these remarkable attributes, our understanding of XVCS remains limited. 
Especially, we have done little about the noise in the reconstructed image up to now.
In this paper, we introduce a novel concept called Static XVCS (SXVCS), which completely eliminates the noise in the reconstructed image. We also demonstrate that the equivalent condition for perfect white pixel reconstruction is simply the existence of SXVCS. For its application, we naturally propose an efficient method for determining the existence of XVCS with perfect white pixel reconstruction. Furthermore, we apply our theorem to $(2,n)$-XVCS and achieve the optimal state of $(2,n)$-XVCS.}




\maketitle

\section{Introduction}\label{sec1}

	\par{Visual cryptography scheme (VCS) was formally proposed by Naor and Shamir in 1994. It encoded a binary secret image (usually denoted by SI) into $n$ transparencies, called shares and each pixel in $S$ was encoded into $m$ sub-pixels in each of the $n$ shares. The share images are printed on the transparent sheets and then distributed secretly to $n$ related participants. If and only if no less than $k$ shares are superimposed, the secret image $S$ is revealed, otherwise, nothing but a random image is obtained. Such a scheme is called a $(k,n)$-VCS. VCSs own the stacking-to-see decryption property such that they are widely used in image encryption \cite{naor1994, Adhikari2014, Drost1996, Wang2011, Bao2017, Jia2018Colla}, visual authentication \cite{1997Auth, Yan2021Auth, 2016Design, 2013Secret, Fu2018} and information hiding \cite{Yan2020, Yuan2014, 2017Singh, Yang2021, Wu2021}.}
	
	
	\par{$(k,n)$-VCSs have a very rich literature since their first appearance in 1994 \cite{naor1994}. However, $(k,n)$-VCSs have strict restrictions on the access structures as they entitle every participant the same privilege, which is not enough to process complex access structures. Therefore, researchers have studied VCSs for general access structures. In 1996, Ateniese et al. constructed two OVCSs for general access structure for the first time. Although the researches on OVCSs with general access structures are far less than $(k,n)$-VCS, its theorem models are deeply studied including finding the optimal combinations of the basis matrices' componentsby using interger linear programming (ILP) \cite{Shyu2015}, linear algebra \cite{Adhikari2014}, the construction by recursively calling (2,2)-VCS \cite{Liu2010}, as well as the construction from combinatorics \cite{Giuse1996}. Among them, the work of Adhikari et al. serves as a cornerstone for our study. Utilizing linear algebra as a fundamental mathematical tool, Adhikari et al. harnessed the idea that the collection of all solutions to a system of linear homogeneous equations over the binary field forms a vector space over the base field. This approach facilitated the development of a VCS construction for specific general access structures.}

	\par{The foundational mathematical operation for the physical implementation of the aforementioned schemes is the Boolean OR operation, which leads to the classification of OR-based VCS (OVCS). However, OVCS suffers from the huge share size (reflected by pixel expansion) and very poor quality (reflected by contrast) of the recovered secret image. Several papers \cite{blundo1999contrast,eisen2002threshold, Shyong2011, Shyu2015} have tried to minimize the pixel expansion and maximize the contrast. Simultaneously, XOR-based VCS (XVCS) has been investigated to attain improved properties like superior contrast and resolution, advancing the comprehension and practicality of VCS for general access structures. By leveraging lightweight devices and emerging technologies, including cell phones, smart devices, and flexible screens, the implementation of XVCS becomes increasingly feasible for diverse practical applications. As a result, XVCS possesses significant potential for extensive adoption in the future.}
	
	\par{In the construction of XVCS, Tuyls et al. \cite{Tuyls2005} proposed threshold XVCS, presenting various $(2, n)$ and $(k, n)$ schemes. Liu et al. \cite{Liu2010} proposed the first general access structure XVCS, employing a $(2, 2)$-XVCS share generation algorithm repeatedly, even though one participant needed to hold multiple shares. Fu et al. \cite{fu2014optimal} proposed a necessary condition for the optimality of pixel expansion in traditional XVCS for general access structures, and they confirmed the existence of a perfect $(n, n)$-XVCS. All these XVCS schemes address non-monotonicity in access structure, where stacking a superset's shares might not disclose the secret. However, given that access structures are typically public, XVCS can be defined to allow qualified subsets of participants to reconstruct the secret image.}

    \par{Notably, Gang Shen et al. \cite{Shen2017} made significant progress in adapting Adhikari's linear algebra method from OVCS to XVCS. By using this technique, they achieved perfect contrast and pixel expansion, and provided sufficient and necessary conditions for this construction. However, constructing such an XVCS requires such strict conditions that few general access structures can meet them. Although the researchers attempted to address this issue through a "region-by-region" construction method, the share of this approach consists of many regions, thus introducing a new form of "pixel expansion" compared to the original image.
    Therefore, there is still much room for improvement in constructing the optimal XVCS using linear systems.
    }
    
    \par{ In addition, we noticed that most previous studies do not foucs on the noise in the reconstructed images. Considering that a noise-free reconstructed image can greatly enhance the effectiveness of secret recovery, the existence of a noise-free construction for VCS, and whether there is a construction that can simultaneously achieve optimal pixel expansion, optimal contrast, and noise-free reconstruction, are indeed worthy new research questions.}

	\subsection{Our Contributions}
	

    \par{In this paper, we extend the methods of previous studies \cite{Adhikari2014}\cite{Shen2017}, constructing a noise-free XVCS, referred to as a Static XOR-based Visual Cryptography Scheme (SXVCS), through two $AX=B$ type linear matrix systems(we sometimes call it as linear equations systems as well since they are so much similar). We further investigate its relationships with other XVCS, as well as its applications in some specific access structures. All the XVCS that we construct in this paper adhere to the basic definition \cite{Tuyls2005, Liu2010} without the need for probabilistic, multi-level grayscale, or multi-region constructions, and each participant only needs to hold one share.
    
    Different from previous studies, our approach focus on eliminating the noise in the reconstructed images of XVCS. Firstly, we examine the methods and conditions for constructing XVCS from two matrix systems. Subsequently, we expand these methods and conditions into the construction based on $2k$ matrix systems. We then demonstrate that an XVCS for each general access structure can be constructed through matrix systems, with the same image reconstruction performance.
    Next, we provide and prove the necessary and sufficient conditions for constructing an SXVCS, i.e., our theorem on the equivalence between a Static XVCS and a Perfect White-reconstructed XVCS (PW-XVCS). Concurrently, using these conclusions, we construct a $(2,n)$-XVCS that achieves optimal pixel expansion, optimal average contrast, and reconsted image without random noise , with a algorithm with time complexity of $O(n\log n)$. All detailed descriptions and proofs can be found in the main text.}

	\subsection{Organization of the paper}
    

    \par{The rest of the paper is organized as follows. In Section \ref{s2}, we provide necessary preliminaries and adopt simplified notations for ease of understanding, concurrently introducing the concepts of SXVCS and SemiSXVCS. In Section \ref{s3}, we present a collection of theorems related to SXVCS, complemented by illustrative examples. These results share similarities with the discussions in \cite{Adhikari2014} and \cite{Shen2017}, albeit approached from a fresh viewpoint. In Section \ref{s4}, we expand the construction of XVCS using $2k$ linear systems, where we establish the equivalence between the existence of SXVCS and SemiSXVCS. We also apply this key theorem to quickly determine whether an access structure possesses a "perfect white pixel reconstruction." In Section \ref{s5}, we addresses the Optimal $(2,n)$-XVCS, referring to a $(2,n)$-SXVCS with optimal pixel expansion and optimal average contrast. In addition, we formulate an efficient algorithm characterized by a time complexity of $O(n\log n)$. In Section \ref{s6}, we present experimental results and comparisons to provide empirical support for the theoretical frameworks discussed in the preceding sections. Finally, in Section \ref{s7}, we conclude the paper and suggest potential directions for future research.}
	
	\par{Figure \ref{diagram} represents the primary structure of the theorems discussed in this paper.}

\begin{figure*}\caption{The diagram of this paper's theorems}\label{diagram}
	\begin{tikzpicture}[box/.style={draw=black,rounded corners},edge from parent/.style={green,thick,draw}]
		\node[box] (1) at(1,0) {\large $2$ Linear Systems};
		\node[box] (11) at(-5,0) {\large  \cite{Adhikari2014} and \cite{Shen2017}};
		\draw[->] (11)--node[pos=0.5,above,sloped]{\scriptsize Add Columns} (1);
		\node[box] (3) at(6,0) {\large  $2k$ Linear Systems};
		\draw[->] (1)--node[pos=0.5,above,sloped]{\scriptsize Add Linear} node[pos=0.5,below,sloped]{\scriptsize Systems} (3);
		
		\node[box] (4) at(1,-2) {\large  SXVCS};
		\node[box] (5) at(6,-2) {\large  All XVCSs};
		\draw[->] (1) -- (4);
		\draw[->] (3) -- (5);
		\draw[->] (4)-- node[pos=0.5,above,sloped]{\scriptsize $\subseteq$} (5);
		\node[box] (9) at(-5,-2) {\large  SXVCS with Perfect White Pixel};
		\node[box] (8) at(1,-4.5) {\large  Semi-SXVCS};
		\node[box] (10) at(-5,-4.5) {\large  XVCS with Perfect White Pixel};
		\node[box] (if) at(6,-4.5) {\large  $k$ Same Systems};
		\node[box] (other) at(6,-6) {\large  Other XVCSs};
		\draw[->] (if)--node[pos=0.5,above,sloped]{\scriptsize No}(other);
		\draw[->] (5)--(if);
		\draw[->] (if)--node[pos=0.5,above,sloped]{\scriptsize Yes}(8); 
		\draw[<->,line width =1.5pt,black] (4)-- node[pos=0.5,above,sloped]{\scriptsize Equivalent}  node[pos=0.5,below,sloped]{\scriptsize in Existance} (8);
		\draw[<->,line width =1.5pt,black] (4)-- node[pos=0.5,above,sloped]{\scriptsize Equivalent}  node[pos=0.5,below,sloped]{\scriptsize in Existance} (9);
		\draw[<->,line width =1.5pt,black] (8)-- node[pos=0.5,above,sloped]{\scriptsize Equivalent}  node[pos=0.5,below,sloped]{\scriptsize in Existance} (10);
        \draw[<->,line width =1.5pt,black] (9) -- node[pos=0.5,above,sloped]{\scriptsize Equivalent}  node[pos=0.5,below,sloped]{\scriptsize in Existance} (10);
    \end{tikzpicture}
\end{figure*}
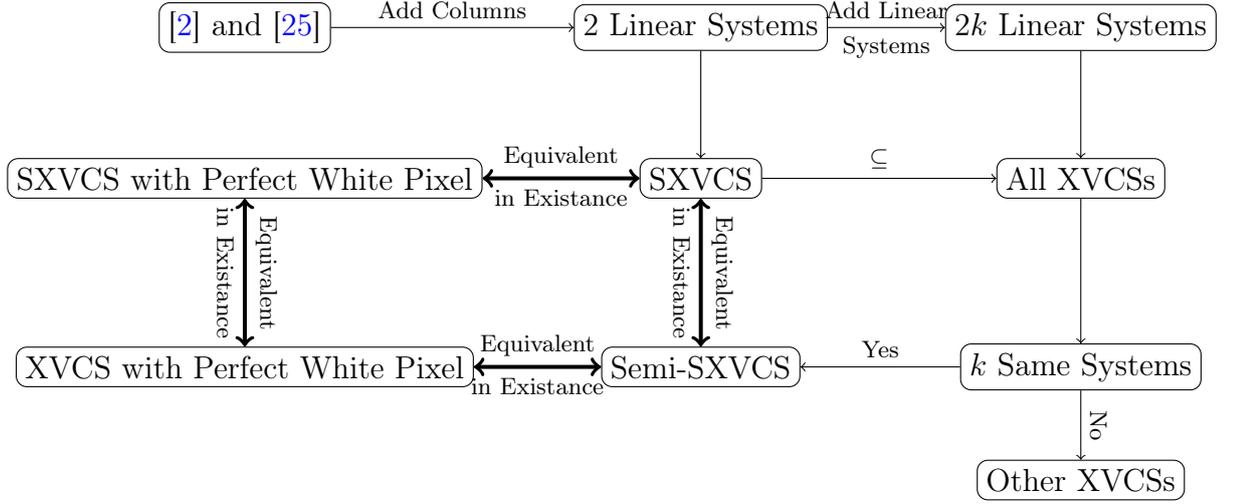



	\section{Preliminaries} \label{s2}
	\par{In this section, we will show the definition of XVCS, SXVCS and SemiSXVCS. Also, some simplified notations will be given. However, the previous work of \cite{Shen2017} and \cite{Adhikari2014} will be written in a brand new form in the next section since that is to some extent very different from their original work.}
	
	\par{In general, VCS relies on an access structure, which is a characterization of the rights of participants. In fact, the concept of access structure is not limited to VCS alone; it is a fundamental concept that applies to all \textit{Secret Sharing Schemes}. If we want to talk about general secret sharing, we need to define access structure first.}
	
	\subsection{General access structure}
	\begin{definition}
		\par{Let $P=\{1,2,...,n\}$ be the \textit{participant set}. Let any given nonempty set $\Gamma_{Qual}\subseteq2^P$ be the \textit{qualified set}. Then we define:
			\begin{itemize}
				\item Let $\Gamma^{-}=\{Q\in\Gamma_{Qual}|\nexists Q'\in \Gamma_{Qual}, s.t.~Q'\subsetneqq Q\}$ be the minimal qualified set, 
				\item Let $\Gamma_{Forb}=\{F\in2^{P}|\forall Q\in \Gamma^{-}, Q\nsubseteq F\}$ be the forbidden set,
				\item And the maximal forbidden set is
				$\Gamma^{+}=\{F\in\Gamma_{Forb}|\nexists F'\in \Gamma_{Forb}, s.t.~F\subsetneqq F'\}$,
				\item We call $\Gamma=(\Gamma_{Qual},\Gamma_{Forb})$ an access structure. Sometimes we change $\Gamma_{Qual}$ for $\Gamma^{-}$ or even just use $\Gamma^{-}$ to represent the whole access structure despite its ambiguity compared to $(\Gamma_{Qual},\Gamma_{Forb})$.
			\end{itemize}
			If $\Gamma^{-}=\{Q\in2^{P}||Q|=k\}$ and $|P|=n$, we specially call it a $(k,n)$ access structure.
		}
	\end{definition}
	
	\par{In this context, we do not focus on the monotonicity of XVCS. Consequently, it is important to bear in mind that a set $Q\in 2^{P}$ satisfying $\exists Q_{0}\in \Gamma^{-}\ s.t. Q_{0}\subsetneqq Q$ does not necessarily imply $Q\in \Gamma_{Qual}$. Furthermore, the set $Q$ is not a subset of $\Gamma_{Forb}$, as individuals can still reconstruct the Secret Image (SI) with shares in $Q_{0}\subsetneqq Q$ by disregarding shares in $Q\setminus Q_{0}$. In general, researchers tend to discuss access structures with monotonicity due to their favorable mathematical properties. However, it is worth noting that many secret sharing schemes, including VCS and particularly XVCS, do not always exhibit monotonic increases. Even when they do, they may fall short in other aspects, such as pixel expansion.}
	
	\par{After providing the definition of the access structure, now we define XVCS.}
	
	\subsection{XOR-based Visual Cryptography Scheme(XVCS)}
	
	\subsubsection{General XVCS}
	\begin{definition}\label{VCS}
		\par{Given an access structure $\Gamma=(\Gamma_{Qual},\Gamma_{Forb})$ with the same notations above. If two (finite) collections of $n\times m$ Boolean matrices $C_{0}$,$C_{1}$ satisfy the following two conditions:
			\begin{center}
				\begin{itemize}
					\item \textbf{Contrast Condition:} $\forall Q\in \Gamma_{Qual}$,$\forall M_{0}\in C_{0},M_{1}\in C_{1}s.t. \omega(\oplus(M_1[Q]))>\omega(\oplus(M_0[Q]))$;
					
					\item \textbf{Security Condition:} $\forall F=\{i_{1},i_{2},\cdots,i_{p}\}\in \Gamma_{Forb}$, the two collections of $p\times n$ matrices $D_{0}$ and $D_{1}$ obtained by restricting $C_{0}$ and $C_{1}$ to rows $i_{1},i_{2},\cdots,i_{p}$, respectively, denoted by $D_{0}=C_{0}[F]$ and $D_{1}=C_{1}[F]$, are distinguishable in the sense that they contain the same matrices with the same frequencies.
				\end{itemize}
			\end{center}
			then we say $C_{0}$ and $C_{1}$ construct an XVCS on $\Gamma$ where $\omega(\mathbf{v})$ means the number of "one"s in the boolean vector $\mathbf{v}$, $M[Q]$ means the matrix $M$ is restricted to rows $i_{1},i_{2},\cdots,i_{p}$($X=\{i_{1},i_{2},\cdots,i_{p}\}$) and $\oplus(M)$ is a boolean row vector which is the result of stacking every row of M through the operation XOR. We call the parameter $m$ the pixel expansion and define contrast of a set $Q\in \Gamma_{Qual}$ as $\displaystyle{\alpha(Q):=
				\dfrac{
					\dfrac{\displaystyle{\sum_{M_{1}\in C_{1}}\omega(\oplus M_{1}[Q])}}
					{|C_{1}|}
					-
					\dfrac{\displaystyle{\sum_{M_{1}\in C_{0}}\omega(\oplus M_{0}[Q])}}
					{|C_{0}|}
				}
				{m}
			}$. The average of all $\alpha(Q)$ is denoted by $\alpha$.
		}
	\end{definition}
	
	\par{In terms of security conditions, we consider $D_{0}$ and $D_{1}$ to be distinguishable if they contain the same matrices with the same frequencies. This notion is not as straightforward as "completely equal." Therefore, we will more frequently utilize Corollary \ref{better_security_condition} in the following sections.}
	
	\begin{corollary}[Better Security Condition]\label{better_security_condition}
		\par{If an XVCS already satisfies \ref{VCS}, there exists another XVCS that meets the simpler security condition $C_{0}[F]=C_{1}[F]$. Furthermore, the two XVCSs have the same impact on the reconstructed image.}
		
		\begin{proof}
			\par{Assume that we have already constructed an XVCS with two collections of basic matrices $C_{0}$ and $C_{1}$, with the number of elements being $|C_{0}|=p_{0}$ and $|C_{1}|=p_{1}$. Let $C_{0}^{*}:=p_{1}C_{0},C_{1}^{*}:=p_{0}C_{1}$, where $p_{0}C_{1}$ and $p_{1}C_{0}$ represent two new collections of matrices defined by $p_{0}$ "$C_{1}$"s and $p_{1}$ "$C_{0}$"s. It is easy to verify that $C_{0}^{*}$ and $C_{1}^{*}$ also construct an XVCS with $|C_{0}^{*}|=|C_{1}^{*}|$ and have the same effect on the reconstructed image.}
		\end{proof}
	\end{corollary}
	
	\par{As mentioned earlier, this paper primarily focuses on SXVCS and SemiSXVCS. Therefore, we will now delve into the definitions and explanations of SXVCS and SemiSXVCS.}
	
	\subsubsection{Static XVCS(SXVCS) and Semistatic XVCS(SemiSXVCS)}
	
	\par{SXVCS is a mathematical description of XVCS without noise in the reconstructed image which is related to the qualified set instead of the forbidden set. So it is natural that we only have to change the contrast condition in \ref{VCS}.}
	\begin{definition}\label{SXVCS}
		\par{Given an access structure $\Gamma=(\Gamma_{Qual},\Gamma_{Forb})$ with the same notations above.If two (finite) collections of $n\times m$ Boolean matrices $C_{0}$,$C_{1}$ satisfy the following two conditions:
			\begin{center}
				\begin{itemize}
					\item \textbf{Static Contrast Condition:} 
					For all $Q \in \Gamma_{Qual}$, and for any $M_{0} \in C_{0}$ and $M_{1} \in C_{1}$, $\oplus(M_0[Q])=:E_{0}(Q)$ and $\oplus(M_1[Q])=:E_{1}(Q)$ are independent with $M_{0}$ and $M_{1}$ and $\omega(E_{0}(Q))<\omega(E_{1}(Q))$. Compared to other XVCS, here the reconstructed pixels $E_{0}(Q)$ and $E_{1}(Q)$ are solely dependent on $Q$ and not influenced by the specific selection of $M_{0}$ and $M_{1}$.
					
					\item \textbf{Security Condition:} $\forall F=\{i_{1},i_{2},\cdots,i_{p}\}\in \Gamma_{Forb}$, the two collections of $p\times n$ matrices $D_{0}$ and $D_{1}$ obtained by restricting $C_{0}$ and $C_{1}$ to rows $i_{1},i_{2},\cdots,i_{p}$, respectively, denoted by $D_{0}=C_{0}[F]$ and $D_{1}=C_{1}[F]$, are distinguishable in the sense that they contain the same matrices with the same frequencies.
				\end{itemize}
			\end{center}
			then we say $C_{0}$ and $C_{1}$ construct an SXVCS on $\Gamma$.
		}
	\end{definition}
	
	\par{Also, we have a seemingly weaker concept called SemiSXVCS.
		\begin{definition}\label{SemiSXVCS}
			\par{Given an access structure $\Gamma=(\Gamma_{Qual},\Gamma_{Forb})$ with the same notations above.If two (finite) collections of $n\times m$ Boolean matrices $C_{0}$,$C_{1}$ satisfy the following two conditions:
				\begin{center}
					\begin{itemize}
						\item \textbf{Semistatic Contrast Condition:} For all $Q \in \Gamma_{Qual}$, and for any $M_{0} \in C_{0}$, $\oplus(M_0[Q])=:E_{0}(Q)$ is independent with $M_{0}$ and $\omega(E_{0}(Q))<\omega(E_{1}(Q))$.
						
						\item \textbf{Security Condition:} $\forall F=\{i_{1},i_{2},\cdots,i_{p}\}\in \Gamma_{Forb}$, the two collections of $p\times n$ matrices $D_{0}$ and $D_{1}$ obtained by restricting $C_{0}$ and $C_{1}$ to rows $i_{1},i_{2},\cdots,i_{p}$, respectively, denoted by $D_{0}=C_{0}[F]$ and $D_{1}=C_{1}[F]$, are distinguishable in the sense that they contain the same matrices with the same frequencies.
					\end{itemize}
				\end{center}
				then we say $C_{0}$ and $C_{1}$ construct a White- SemiSXVCS on $\Gamma$. Similarly, we have Black- SemiSXVCS.
			}
		\end{definition}
	}
	\par{It is easy to check that XVCS with pixel expansion $1$ is SXVCS. Despite what it looks like, we will prove the fact that the existence of a SemiSXVCS implies the existence of an SXVCS.}
	
	\par{For SemiSXVCS, we can check that an XVCS with $\oplus M_{0}[Q]\equiv\mathbf{0}(\forall Q\in \Gamma_{Qual},\forall M_{0}\in C_{0}$) is a SemiSXVCS, which we call PW-XVCS. In \cite{2005Tuyls}, Tuyls constructed $(2,n)$-XVCS with the property above(thus is a SemiSXVCS) by using binary code. That means SemiSXVCS is not that rare.}
	
	\par{Here, we present the reconstructed images of an SXVCS and a SemiSXVCS as shown in \ref{fig:example1}. Given the access structure as (2,3) access structure. The reconstructed image of SXVCS clearly exhibits a noise-free quality, with all subpixels arranged in a neat manner, resulting in the vertical lines observed in the image.}
	
	\begin{figure}[htbp]
		\centering
		\subfloat[SemiSXVCS]{\includegraphics[width=2.5in]{"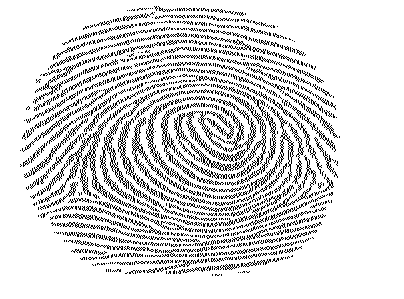"}\label{fig:semi12}} \hspace{0.03in}
         \subfloat[SXVCS]{\includegraphics[width=2.5in]{"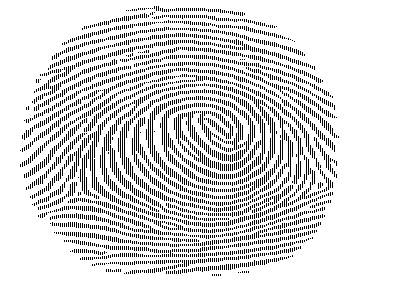"}\label{fig:semi12}} \hspace{0.03in}
		\caption{Comparison of SemiSXVCS(PW-XVCS) and SXVCS reconstruction image examples}
		\label{fig:example1}
	\end{figure}

	\subsection{The Matrix Formulation of Access Structure}

    \par{The mathematical tools employed in this paper primarily rely on linear algebra. In order to effectively describe the theorems related to linear system, it is essential to represent the access structure in matrix form. This approach proves valuable as it eliminates the need for complex transformations between qualified sets, subsets, vectors, and matrices. By utilizing a matrix formulation for the access structure, we can avoid convoluted terminology and present a more coherent and structured representation.}
 
	\par{First we have to introduce the concept of \textit{bitset}.
		
		\begin{definition}
			\par{Given a boolean vector $\textbf{V}$, its bitset $V$ is defined by $n\in V \iff$ the $n$th component of $\textbf{V}$ is "$1$". It is clear that one boolean vector has and only has one bitset.}
		\end{definition}
		
		\par{With the definition of bitset, it is clear that every subset $Q$ of qualified set $\Gamma_{Qual}$ is corresponding with a unique boolean row vector. For convenience, we denote it by $\mathbf{Q}$.}
		
		\begin{definition}\label{MinQuaMatr}
			\par{Given a qualified set $\Gamma_{Qual}=\{Q_{1},Q_{2},\cdots,Q_{t}\}$, we call the $t\times n$ matrix
				\begin{center}
					$$\left[\begin{array}{c}
						\mathbf{Q_{1}}\\
						\mathbf{Q_{2}}\\
						\vdots\\
						\mathbf{Q_{t}}\\
					\end{array}\right]$$
				\end{center}
				
				the \textit{Qualified Matrix} of $\Gamma_{Qual}$. For convenience, we denote it by $\mathbf{\Gamma_{Qual}}$}
		\end{definition}
		
		\par{For example, if $\Gamma_{Qual}=\{\{1,2\},\{1,3\},\{1,4\}\}$, then we have
			$$\mathbf{\Gamma_{Qual}}=\left[\begin{array}{cccc}
				1&	1&	0&	0\\
				1&	0&	1&	0\\
				1&	0&	0&	1
			\end{array}
			\right].$$}
		\par{No wonder a qualified set usually has more than one qualified matrices. Since these matrices are the same up to row permutation and they all correspond with only one unique qualified set, we usually choose one suitable matrix before we begin a proof or something else.}
		
		\par{Similarly, We can also define \textit{Minimal Qualified matrix, Forbidden Matrix} and \textit{Maximal Forbidden Matrix}. }

        \section{XVCS Construction via $2$ linear equations systems}\label{s3}
		\par{Before our research, \cite{Adhikari2014} and \cite{Shen2017} have already given some useful and important theorems. Firstly, Adhikari introduced his theorem of constructing some OVCSs through $2$ linear systems on a binary field. Then, in 2017 Gang Shen et al. change Adhikari's theorem for XVCS and gave a useful equivalent condition of whether an access structure has an XVCS with pixel expansion $1$. That is, the $2$ solution sets(we will denote it by $S_{0}$ and $S_{1}$ to distinguish from other XVCS constructed from other methods of which basis matrices are denoted by $C_{0}$ and $C_{1}$) of the following $2$ linear systems on binary field constructs an XVCS with pixel expansion $1$ if and only if an odd number of rows $\mathbf{\Gamma_{Qual}}$ stacked by XOR is still a row of $\mathbf{\Gamma_{Qual}}$ and even number of rows of $\mathbf{\Gamma_{Qual}}$ stacked by XOR(denoted by $\mathbf{V}$) is zero or $V$ cannot be a subset of any element of the minimal qualified set. 
			\begin{center}
				$\mathbf{\Gamma_{Qual}}X=\mathbf{0_{t,1}}$;\\
				$\mathbf{\Gamma_{Qual}}X=\mathbf{1_{t,1}}$
			\end{center}
			Here in the above two linear systems on binary field, $\mathbf{1_{t,1}}$ represents a column vector sized $t$ with each component equals to "$1$" and $t:=|\Gamma_{Qual}|$.
		}
		\par{Now, It is clear why this method is disabled for constructing XVCS with minimal pixel expansion larger than one: the solutions of these linear systems are in fact matrices of size $n\times 1$ instead of the more generally $n\times m$. 
        To solve this problem, we can expand $X$, $\mathbf{0_{t,1}}$, and $\mathbf{1_{t,1}}$ to matrices with $m$ columns, as shown below. Despite the apparent formula difference, the result remains as two linear systems.
			\begin{center}
				\begin{eqnarray}
					\mathbf{\Gamma_{Qual}}X=B_{0};\label{white_SXVCS}\\
					\mathbf{\Gamma_{Qual}}X=B_{1}\label{black_SXVCS},
				\end{eqnarray}
			\end{center}
            Here $B_{0},B_{1}$ are all boolean matrices of size $t\times m$ and $X$ is of size $n\times m$.
		}

        \par{One thing that needs to be clarified is that the parameter $m$ represents the pixel expansion of the desired outcome. It should be noted that simply writing down the two systems and solving them does not guarantee the immediate formation of an XVCS. In the following discussion, we will examine the conditions under which the two linear systems can construct an XVCS. Specifically, we will address the contrast condition and the security condition.}

		\subsection{Contrast condition}
		\par{Suppose that we have already constructed an XVCS no matter how did we get it. Put $M_{0}\in C_{0}$ into the place of $X$ in \eqref{white_SXVCS} and you will find that the left side of the equation simply means the stack of shares in $\Gamma_{Qual}$ and the right side of the equation is simply the result of these stack of shares in $\Gamma_{Qual}$ which is realized by block matrix multiplication. For example, if a $(2,3)$-XVCS has a matrix $[\mathbf{0_{3,1}}\ \mathbf{1_{3,1}}]$ in $C_{0}$, then it should look like this:
			\begin{center}
				$\left[\begin{array}{ccc}
					1&	1&	0\\
					1&	0&	1\\
					0&	1&	1
				\end{array}
				\right]\left[\begin{array}{cc}
					0&	1\\
					0&	1\\
					0&	1
				\end{array}
				\right]=\left[\begin{array}{cc}
					0&	0\\
					0&	0\\
					0&	0
				\end{array}
				\right].$
			\end{center}
		}
		
		\par{After realizing this, the contrast condition is reduced to check $B_{0}$ and $B_{1}$ in \eqref{white_SXVCS} and \eqref{black_SXVCS}. Also, notice that if an XVCS is constructed from $S_{0}$ and $S_{1}$ which are the solution sets of \eqref{white_SXVCS} and \eqref{black_SXVCS}, it satisfies the static contrast condition in \ref{SXVCS} so it is an SXVCS. So we conclude the result as the following lemma.
  
			\begin{lemma}\label{contrast_SXVCS}
				Suppose that \eqref{white_SXVCS} and \eqref{black_SXVCS} are consistant. If $\forall k\in\{1,2,\cdots,t\}, \omega(B_{0}[\{k\}])<\omega(B_{1}[\{k\}])$, then $S_{0}$ and $S_{1}$ satisfies the static contrast condition in \ref{SXVCS}.
			\end{lemma}
		}
		\subsection{Security Condition}
		\par{It is easy to check that \ref{contrast_SXVCS} is not enough for the security condition. For example, check this one: given the qualified matrix as $\mathbf{\Gamma_{Qual}}=\left[\begin{array}{cccc}
				1&1&1&0\\
				1&1&0&1\\
				0&1&1&1
			\end{array}\right]$ and expect an XVCS of pixel expansion $1$. We can write down the $2$ systems like this:
			\begin{center}
				$\left[\begin{array}{cccc}
					1&1&1&0\\
					1&1&0&1\\
					0&1&1&1
				\end{array}\right]X=\mathbf{0_{3,1}}$\\
				$\left[\begin{array}{cccc}
					1&1&1&0\\
					1&1&0&1\\
					0&1&1&1
				\end{array}\right]X=\mathbf{1_{3,1}}$
			\end{center}
			and we can get $S_{0}$ and $S_{1}$:
			\begin{center}
				$S_{0}=
				\left\{\left[
				\begin{array}{c}
					0\\
					0\\
					0\\
					0\\
				\end{array}
				\right],\left[
				\begin{array}{c}
					1\\
					0\\
					1\\
					1\\
				\end{array}
				\right]\right\},
				S_{1}=
				\left\{\left[
				\begin{array}{cc}
					0\\
					1\\
					0\\
					0\\
				\end{array}
				\right],\left[
				\begin{array}{c}
					1\\
					1\\
					1\\
					1\\
				\end{array}
				\right]\right\}.$
			\end{center}
           Here it becomes evident that the second participant violates the security condition. Therefore, it is crucial to determine the equivalent condition under which $S_{0}$ and $S_{1}$ can construct an XVCS (or equivalently an SXVCS). Fortunately, the solution to this problem has already been proposed by \cite{Adhikari2014} and \cite{Shen2017}. Our task now is to expand upon their ideas.
            }
		
		\par{Now we denote the solution set of the following linear systems as $S_{w}$.
			\begin{equation}\label{perfect_white_SXVCS}
				\mathbf{\Gamma_{Qual}}X=\textbf{0}_{t,m}.
			\end{equation}	
			
		}
		\par{Similar to \cite{Adhikari2014} and \cite{Shen2017}, we proof the following lemmas.}
		
		\begin{lemma}\label{lem01}
			\par{Suppose that $\mathbf{\Gamma_{Qual}}X=B_{i}$ has a particular solution $X_{i}(i=0,1)$.Then the solution set is $S_{i}=\{X|\exists X'\in S_{w}, X=X'+X_{i}\}=:S_{w}+X_{i}(i=1,2)$.}
		\end{lemma}
		\par{since the proof of this lemma is written on almost every "Linear Algebra" teaching material, it is omitted.}
		
		\begin{lemma}\label{security_theorem}
			\par{$S_{0}[F]=S_{1}[F]\Longleftrightarrow \exists X_{0}\in S_{0},X_{1}\in S_{1}\ s.t.\ X_{0}[F]=X_{1}[F]$. And if $\forall F\in \Gamma_{Forb},S_{0}[F]=S_{1}[F]$, $S_{0}$ and $S_{1}$ satisfies the security condition in \ref{SXVCS}.}
		\end{lemma}
		
		\begin{proof}		
            "$\Rightarrow$" It is obvious since $S_{0}[F]=S_{1}[F]$.

            "$\Leftarrow$" By lemma \ref{lem01}, $S_{0}[F]=\{X|\exists X'\in S_{w}, X=X'+X_{0}\}[F]=\{X|\exists X'[F]\in S_{w}[F], X=X'[F]+X_{0}[F]\}=\{X|\exists X'[F]\in S_{w}[F], X=X'[F]+X_{1}[F]\}=\{X|\exists X'\in S_{w}, X=X'+X_{1}\}[F]=S_{1}[F]$. In short, $S_{0}[F]=S_{w}[F]+X_{0}[F]=S_{w}[F]+X_{1}[F]=S_{1}[F]$.
		\end{proof}
		
		\par{Please pay attention! We will next give an important corollary \ref{insertion_SXVCS}, which will show more of its importance in the next section.}
		
		\begin{corollary}[SXVCS Insertion Theorem]\label{insertion_SXVCS}
			\par{Suppose that we have already constructed an SXVCS with basis matrices $C_{0}$ and $C_{1}$. Then there exists an SXVCS constructed from \eqref{white_SXVCS} and \eqref{black_SXVCS} such that they have the same effect on the reconstructed picture.}
		\end{corollary}
		\begin{proof}
			\par{Choose $M_{0}\in C_{0}$ and $M_{1}\in C_{1}$ randomly. Define $B_{0}:=\mathbf{\Gamma_{Qual}}M_{0}$ and $B_{1}:=\mathbf{\Gamma_{Qual}}M_{1}$. It is clear that no matter which $M_{0}$ and $M_{1}$ you choose,$B_{0}$ and $B_{1}$ are always the same $B_{0}$ and $B_{1}$ since it is an SXVCS. Now consider the following $2$ linear systems:
				$$\mathbf{\Gamma_{Qual}}X=B_{0}$$		
				$$\mathbf{\Gamma_{Qual}}X=B_{1}$$
				Denote the solution sets of them by $S_{0}$ and $S_{1}$. Then we have $C_{0}\subseteq S_{0}$ and $C_{1}\subseteq S_{1}$(if $C_{0}$ and $C_{1}$ are reduced from multisets to sets by deleting the repetitive elements). So by the security condition of the given SXVCS, for all $F\in \Gamma_{Forb}$, we can find $M_{0}^{F}\in C_{0}\subseteq S_{0}$ and $M_{1}^{F}\in C_{1}\subseteq S_{1}$ such that $M_{0}^{F}[F]=M_{1}^{F}[F]$. By lemma \ref{security_theorem}, $S_{0}[F]=S_{1}[F]$.
			}
			\par{For the contrast condition, it is clear that $B_{0}$ and $B_{1}$ work the same with the given SXVCS. So $S_{0}$ and $S_{1}$ construct an SXVCS with the same effect on the reconstructed picture.}
		\end{proof}
		
		\par{This corollary reveals that all access structures with SXVCS can construct SXVCS from $2$ linear systems and the linear algebra is much simpler. }
		
		\par{Although lemma \ref{security_theorem} has already given the equivalent condition of $S_{0}$ and $S_{1}$ satisfying the security condition, it is far from being useful for cumulating. The next is the matrix formula of lemma \ref{security_theorem}. Its proof is omitted since it directly comes from the block matrix multiplication. }
		
		\begin{corollary}[SXVCS General Security Theorem]\label{general_security_theorem_mat_SXVCS}
			\par{$S_{0}[F]=S_{1}[F]\Longleftrightarrow $ the following linear systems is consistent.
				\begin{equation}\label{general_security_blkmat_formula}
					\left[\begin{array}{cc}
						\mathbf{\Gamma_{Qual}}&	 \mathbf{0_{t,n}} \\
						\mathbf{0_{t,n}}&		 \mathbf{\Gamma_{Qual}} \\
						T_{F}&					 T_{F}
					\end{array}\right]X=
					\left[\begin{array}{c}
						B_{0} \\
						B_{1} \\	
						\mathbf{0_{|F|,m}}
					\end{array}\right],
			\end{equation}}
			$T_{F}$ is defined similarly to \cite{Adhikari2014}. The $i$th column of $T_{F}$ is $\mathbf{0_{|F|,1}} \Leftrightarrow$ the $i$th component of $\mathbf{F}$ is "$0$". The $i$th column of $T_{F}$ is $j$th column of the identity matrix $I_{|F|}\Leftrightarrow$ the $i$th component of $\mathbf{F}$ is the $j$th component of value "$1$".
		\end{corollary}
		
		\par{When $B_{0}=\mathbf{0}_{t,m}$, the result is reduced to the following corollary.}
		
		\begin{corollary}[Perfect White Pixel SXVCS General Security Theorem]\label{security_perfect_white_pixel_SXVCS}
			\par{	When $B_{0}=\mathbf{0}_{t,m},\ S_{0}[F]=S_{1}[F]\Longleftrightarrow $ the following linear systems is consistent.
				\begin{equation}\label{perfect_white_security_blkmat_formula}
					\left[\begin{array}{c}
						\mathbf{\Gamma_{Qual}}\\
						T_{F}\\
					\end{array}\right]X=
					\left[\begin{array}{c}
						B_{1}\\
						\mathbf{0}_{|F|,m} \\
					\end{array}\right]
			\end{equation}}
		\end{corollary}
		\begin{proof}
			\par{We only need to check that \eqref{general_security_blkmat_formula} and \eqref{perfect_white_security_blkmat_formula} are consistent or inconsistent simultaneously. }
			\par{When \eqref{general_security_blkmat_formula} has a particular solution $X_{0}=\left[\begin{array}{c}
					X_{1}\\
					X_{2}
				\end{array}\right]$ where both $X_{1},X_{2}$ are of size $n\times m$. Then it is easy to check that $X_{1}+X_{2}$ is a solution to \eqref{perfect_white_security_blkmat_formula}.}
			\par{Conversely, if \eqref{perfect_white_security_blkmat_formula} has a particular solution $X_{0}$, then $\left[\begin{array}{c}
					\mathbf{0_{n,m}}\\
					X_{0}
				\end{array}\right]$ is a solution to \eqref{general_security_blkmat_formula}. The proof is thus finished.}
		\end{proof}
		
		\par{The reason for highlighting the condition of $B_{0}=\mathbf{0_{t,m}}$(i.e. perfect white pixel reconstruction) will be explained in next subsection. Now we just point out that \cite{Shen2017}'s work is just under the condition of $B_{0}=\mathbf{0_{t,1}}$.}
		
		\par{After these many lemmas and corollaries, we show one example as follows.}
		
		\begin{example}
			\par{Given the qualified matrix(it is actually a $(2,3)$ scheme): $$\mathbf{\Gamma_{Qual}}=\left[
				\begin{array}{ccc}
					1&1&0\\
					1&0&1\\
					0&1&1\\
				\end{array}
				\right],$$
				Now, with the expectation of pixel expansion $2$, construct a $(2,3)$-SXVCS.
			}
			\par{Firstly, we write down the systems like this: 
				\begin{eqnarray*}
					\left[
					\begin{array}{ccc}
						1&1&0\\
						1&0&1\\
						0&1&1\\
					\end{array}
					\right]
					X&=&
					\left[
					\begin{array}{cc}
						0&0\\
						0&0\\
						0&0\\
					\end{array}
					\right],\label{eq08}\\		
					\left[
					\begin{array}{ccc}
						1&1&0\\
						1&0&1\\
						0&1&1\\
					\end{array}
					\right]
					X&=&
					\left[
					\begin{array}{cc}
						1&0\\
						0&1\\
						1&1\\
					\end{array}
					\right],\label{eq09}
				\end{eqnarray*}
			}
			
			\par{One thing you may wonder is that why $B_{0}$ is defined as $\mathbf{0_{3,2}}$ and $B_{1}$ is defined as $\left[
				\begin{array}{cc}
					1&0\\
					0&1\\
					1&1\\
				\end{array}
				\right]$. However, it is too complex to explain at this point. We will give a useful algorithm with complexity $O(n)$ for $(2,n)$-SXVCS later but we still have not found an effective algorithm for general SXVCS. Here, just check the contrast condition by lemma \ref{contrast_SXVCS} first.
			}
			\par{Now we consider the security condition. It's obvious that $\mathbf{\Gamma_{Forb}}=\left[
				\begin{array}{ccc}
					1&0&0\\
					0&1&0\\
					0&0&1
				\end{array}
				\right]$. So by \ref{general_security_theorem_mat_SXVCS}, we have
				\begin{eqnarray*}
					\left[
					\begin{array}{ccc:ccc}
						1&1&0&0&0&0\\
						1&0&1&0&0&0\\
						0&1&1&0&0&0\\ \hdashline[1pt/1pt]
						0&0&0&1&1&0\\
						0&0&0&1&0&1\\
						0&0&0&0&1&1\\ \hdashline[1pt/1pt]
						1&0&0&1&0&0\\
					\end{array}
					\right]X&=&
					\left[
					\begin{array}{cc}
						0&0\\
						0&0\\
						0&0\\ \hdashline[1pt/1pt]
						1&0\\
						0&1\\
						1&1\\ \hdashline[1pt/1pt]
						0&0\\
					\end{array}
					\right],\label{eq10}\\
					\left[
					\begin{array}{ccc:ccc}
						1&1&0&0&0&0\\
						1&0&1&0&0&0\\
						0&1&1&0&0&0\\ \hdashline[1pt/1pt]
						0&0&0&1&1&0\\
						0&0&0&1&0&1\\
						0&0&0&0&1&1\\ \hdashline[1pt/1pt]
						0&1&0&0&1&0\\
					\end{array}
					\right]X&=&
					\left[
					\begin{array}{cc}
						0&0\\
						0&0\\
						0&0\\ \hdashline[1pt/1pt]
						1&0\\
						0&1\\
						1&1\\ \hdashline[1pt/1pt]
						0&0\\
					\end{array}
					\right],\label{eq11}\\
					\left[
					\begin{array}{ccc:ccc}
						1&1&0&0&0&0\\
						1&0&1&0&0&0\\
						0&1&1&0&0&0\\ \hdashline[1pt/1pt]
						0&0&0&1&1&0\\
						0&0&0&1&0&1\\
						0&0&0&0&1&1\\ \hdashline[1pt/1pt]
						0&0&1&0&0&1\\
					\end{array}
					\right]X&=&
					\left[
					\begin{array}{cc}
						0&0\\
						0&0\\
						0&0\\ \hdashline[1pt/1pt]
						1&0\\
						0&1\\
						1&1\\ \hdashline[1pt/1pt]
						0&0\\
					\end{array}
					\right].\label{eq12}
				\end{eqnarray*}
				If they are all consistent, then the solution set of the $2$ linear systems given above can construct an $(2,3)$-SXVCS.
			}
			
			\par{We can use \ref{security_perfect_white_pixel_SXVCS} as well since it is simpler.}
			\begin{eqnarray*}
				\left[
				\begin{array}{ccc}
					1&1&0\\
					1&0&1\\
					0&1&1\\ \hdashline[1pt/1pt]
					1&0&0\\
				\end{array}
				\right]
				X&=&
				\left[
				\begin{array}{cc}
					1&0\\
					0&1\\
					1&1\\ \hdashline[1pt/1pt]
					0&0\\
				\end{array}
				\right]\label{eq13}\\
				\left[
				\begin{array}{ccc}
					1&1&0\\
					1&0&1\\
					0&1&1\\ \hdashline[1pt/1pt]
					0&1&0\\
				\end{array}
				\right]
				X&=&
				\left[
				\begin{array}{cc}
					1&0\\
					0&1\\
					1&1\\ \hdashline[1pt/1pt]
					0&0\\
				\end{array}
				\right],\label{eq14}\\
				\left[
				\begin{array}{ccc}
					1&1&0\\
					1&0&1\\
					0&1&1\\ \hdashline[1pt/1pt]
					0&0&1\\
				\end{array}
				\right]
				X&=&
				\left[
				\begin{array}{cc}
					1&0\\
					0&1\\
					1&1\\ \hdashline[1pt/1pt]
					0&0\\
				\end{array}
				\right].\label{eq15}
			\end{eqnarray*}
			\par{Since these $3$ systems given by \ref{security_perfect_white_pixel_SXVCS} are really consistent, the following sets can construct an SXVCS. 
				$$S_{0}=\left\{\left[\begin{array}{cc}
					1&	1\\
					1&	1\\
					1&	1
				\end{array}\right],\left[\begin{array}{cc}
					0&	1\\
					0&	1\\
					0&	1
				\end{array}\right],\left[\begin{array}{cc}
					1&	0\\
					1&	0\\
					1&	0
				\end{array}\right],\left[\begin{array}{cc}
					0&	0\\
					0&	0\\
					0&	0
				\end{array}\right]\right\}$$
				$$S_{1}=\left\{\left[\begin{array}{cc}
					1&	1\\
					0&	1\\
					1&	0
				\end{array}\right],\left[\begin{array}{cc}
					0&	1\\
					1&	1\\
					0&	0
				\end{array}\right],\left[\begin{array}{cc}
					1&	0\\
					0&	0\\
					1&	1
				\end{array}\right],\left[\begin{array}{cc}
					0&	0\\
					1&	0\\
					0&	1
				\end{array}\right]\right\}$$
			}
			
		\end{example}
		
		\subsection{SXVCS with perfect white pixel reconstruction}
		\par{
			It seems easier to construct an SXVCS with its white pixel perfectly reconstructed than to construct an SXVCS with unknown $B_{0}$. Here we will prove that if we have already had an SXVCS, then we can construct an SXVCS with perfect white pixel reconstruction and the same pixel expansion.
			\begin{lemma}
				\par{Given an access structure. Then an SXVCS of pixel expansion $m$ indicates that an SXVCS of the same pixel expansion but with perfect white pixel reconstruction exists.}
			\end{lemma}
			\begin{proof}
				\par{By \ref{insertion_SXVCS}, we suppose that the SXVCS is constructed from $S_{0}$ and $S_{1}$ which are the solution sets of \eqref{white_SXVCS} and \eqref{black_SXVCS}. Then, consider these $2$ sets: $S_{w}$ and $S_{1}+S_{0}:=\{X|X=X_{1}+X_{0},X_{1}\in S_{1},X_{0}\in S_{0}\}$. In fact, $S_{1}+S_{0}=S_{w}+(X_{1}+X_{0})(X_{1}\in S_{1},X_{0}\in S_{0})$ and is the solution set of $\mathbf{\Gamma_{Qual}}X=B_{0}+B_{1}$. No wonder the contrast condition has already been satisfied by the $2$ sets since $\omega(B_{0}[\{k\}])<\omega(B_{1}[\{k\}])(\forall k=1,2,\cdots,t)$. And for all $F\in\Gamma_{Forb}$, $S_{0}[F]=S_{1}[F]$ so $\mathbf{0_{|F|,m}}\in (S_{1}+S_{0})[F]$ and $\mathbf{0_{|F|,m}}\in S_{w}[F]$. By lemma \ref{security_theorem}, the security condition fits. So the set $S_{w}$ and $S_{0}+S_{1}$ construct an SXVCS of the same pixel expansion but with perfect white pixel reconstruction.}
			\end{proof}
		}
		\par{It may look disappointing since the lemma indicates that an access structure with an SXVCS on it must at least have an XVCS with perfect white pixel reconstruction on it. However, with deeper research, we have already proved that it is equivalent to say "there exists an SXVCS" and "there exists an XVCS with perfect white pixel".}
		
		\par{Although SXVCS and SXVCS with perfect white pixel are equivalent in existence, SXVCS is not always an SXVCS with perfect white pixel reconstruction. Still take $(2,3)$-SXVCS as an example:}
		\begin{example}
			\par{Given the qualified matrix: $$\mathbf{\Gamma_{Qual}}=\left[
				\begin{array}{ccc}
					1&1&0\\
					1&0&1\\
					0&1&1\\
				\end{array}
				\right],$$
				Now, with the expectation of pixel expansion $3$, construct a $(2,3)$-SXVCS.
			}
			\par{Firstly, we write down the systems by contrast condition like this: 
				\begin{eqnarray*}
					\left[
					\begin{array}{ccc}
						1&1&0\\
						1&0&1\\
						0&1&1\\
					\end{array}
					\right]
					X&=&
					\left[
					\begin{array}{ccc}
						0&0&1\\
						0&0&0\\
						0&0&1\\
					\end{array}
					\right],\label{eq08}\\		
					\left[
					\begin{array}{ccc}
						1&1&0\\
						1&0&1\\
						0&1&1\\
					\end{array}
					\right]
					X&=&
					\left[
					\begin{array}{ccc}
						1&0&1\\
						0&1&0\\
						1&1&1\\
					\end{array}
					\right],\label{eq09}
				\end{eqnarray*}
			}
			
			\par{Now we consider the security condition. It's obvious that $\mathbf{\Gamma_{Forb}}=\left[
				\begin{array}{ccc}
					1&0&0\\
					0&1&0\\
					0&0&1
				\end{array}
				\right]$. So by \ref{general_security_theorem_mat_SXVCS}, we have
				\begin{eqnarray*}
					\left[
					\begin{array}{ccc:ccc}
						1&1&0&0&0&0\\
						1&0&1&0&0&0\\
						0&1&1&0&0&0\\ \hdashline[1pt/1pt]
						0&0&0&1&1&0\\
						0&0&0&1&0&1\\
						0&0&0&0&1&1\\ \hdashline[1pt/1pt]
						1&0&0&1&0&0\\
					\end{array}
					\right]X&=&
					\left[
					\begin{array}{ccc}
						0&0&1\\
						0&0&0\\
						0&0&1\\ \hdashline[1pt/1pt]
						1&0&1\\
						0&1&0\\
						1&1&1\\ \hdashline[1pt/1pt]
						0&0&0\\
					\end{array}
					\right],\label{eq10}\\
					\left[
					\begin{array}{ccc:ccc}
						1&1&0&0&0&0\\
						1&0&1&0&0&0\\
						0&1&1&0&0&0\\ \hdashline[1pt/1pt]
						0&0&0&1&1&0\\
						0&0&0&1&0&1\\
						0&0&0&0&1&1\\ \hdashline[1pt/1pt]
						0&1&0&0&1&0\\
					\end{array}
					\right]X&=&
					\left[
					\begin{array}{ccc}
						0&0&1\\
						0&0&0\\
						0&0&1\\ \hdashline[1pt/1pt]
						1&0&1\\
						0&1&0\\
						1&1&1\\ \hdashline[1pt/1pt]
						0&0&0\\
					\end{array}
					\right],\label{eq11}\\
					\left[
					\begin{array}{ccc:ccc}
						1&1&0&0&0&0\\
						1&0&1&0&0&0\\
						0&1&1&0&0&0\\ \hdashline[1pt/1pt]
						0&0&0&1&1&0\\
						0&0&0&1&0&1\\
						0&0&0&0&1&1\\ \hdashline[1pt/1pt]
						0&0&1&0&0&1\\
					\end{array}
					\right]X&=&
					\left[
					\begin{array}{ccc}
						0&0&1\\
						0&0&0\\
						0&0&1\\ \hdashline[1pt/1pt]
						1&0&1\\
						0&1&0\\
						1&1&1\\ \hdashline[1pt/1pt]
						0&0&0\\
					\end{array}
					\right].\label{eq12}
				\end{eqnarray*}
				Since they are all consistent, then the solution set of the $2$ linear systems given above can construct an $(2,3)$-SXVCS.
			}
		\end{example}
		
		
		
		
		
	
    \section{XVCS Construction via $2k$ linear equations systems}\label{s4}
	\subsection{The $2k$ Linear systems}

    \par{In Section \ref{s3}, we have already demonstrated that, provided they satisfy Lemma \ref{lem01} and \ref{contrast_SXVCS}, \eqref{white_SXVCS} and \eqref{black_SXVCS} can be used to construct an SXVCS. The question arises as to why the construction is SXVCS instead of general XVCS. In fact, for an XVCS with basis matrices $C_{0}$ and $C_{1}$, if we randomly choose $M_{0},M_{0}'\in C_{0}$, then $\mathbf{\Gamma_{Qual}}M_{0}$ and $\mathbf{\Gamma_{Qual}}M_{0}'$ are not always equal with each other unless it is an SXVCS or at least a SemiSXVCS. Given this observation, it is a natural progression to extend the theorem of construction by $2$ linear systems to encompass $2k$ linear systems.}
	
	\par{Here is the $2k$ linear systems:
		\begin{eqnarray}
			\mathbf{\Gamma_{Qual}}X=B_{01},\ \mathbf{\Gamma_{Qual}}X=B_{02},\ ...,\ \mathbf{\Gamma_{Qual}}X=B_{0k}\label{w_eqs}\\
			\mathbf{\Gamma_{Qual}}X=B_{11},\ \mathbf{\Gamma_{Qual}}X=B_{12},\ ...,\ \mathbf{\Gamma_{Qual}}X=B_{1k}\label{b_eqs}
		\end{eqnarray}
		$B_{ij}$ is boolean matrix of size $t\times m(i=0,1;j=1,2,\cdots,k)$. $X$ is a boolean matrix of size $n\times m$. $t=|\Gamma_{Qual}|$. We denote the solution set of $\mathbf{\Gamma_{Qual}}X=B_{ij}$ as $S_{ij}(i=0,1;j=1,2,\cdots,k)$ and denote the multiset $\coprod_{j=1}^{k}S_{ij}$ as $S_{i}(i=0,1)$. 
	}
	
	\par{Similar to what we have done in Section \ref{s3}, we give the equivalent condition of $S_{0}$ and $S_{1}$ constructing an XVCS.}
	
	\par{By \ref{security_theorem}, it is easy to prove lemma \ref{security_theorem_complete} as follows.}
	\begin{lemma}\label{security_theorem_complete}
		$S_{0}[F]=S_{1}[F]\Leftrightarrow$ there exists bijection $f_{F}:\{1,2,3,...,k\}\rightarrow \{1,2,3,...,k\}\ s.t.\ \forall i \in \{1,2,3,...,k\},\exists X_{0i}\in S_{0i},X_{1f_{F}(i)}\in S_{1f_{F}(i)}\ s.t.\ X_{0i}[F]=X_{1f_{F}(i)}[F]$
		
	\end{lemma}
	\begin{proof}
		\par{\begin{itemize}
				\item["$\Leftarrow$"] By \ref{security_theorem}, $S_{0j}[F]=S_{1f_{F}(j)}[F](\forall j=1,2,\cdots,k)$. So it is clear that $S_{0}[F]=S_{1}[F]$ as a result.
				\item["$\Rightarrow$"] Since $S_{0}[F]=S_{1}[F]$, for any $X_{01}\in S_{01}\subseteq S_{0}$, there exists $X_{1f_{F}(1)}\in S_{1f_{F}(1)}\subseteq S_{1}\ s.t. X_{01}[F]=X_{1f_{F}(1)}[F]$. By lemma \ref{security_theorem}, $S_{01}[F]=S_{1f_{F}(1)}[F]$. Then, $(S_{0}\setminus S_{01})[F]=(S_{1}\setminus S_{1f_{F}(1)})[F]$. Similar to what we have done just now, we can select $X_{02}$ from $S_{02}\subseteq S_{0}\setminus S_{01}$ and find $X_{1f_{F}(2)}\in S_{1f_{F}(2)}\subseteq S_{1}\setminus S_{1f_{F}(1)}\ s.t. X_{02}[F]=X_{1f_{F}(2)}[F]$. Then again we have $(S_{0}\setminus \coprod_{j=1}^{2}S_{0j})[F]=(S_{1}\setminus \coprod_{j=1}^{2}S_{1f_{F}(j)})[F]$. Repeat this procedure for $k$ times and we will obtain a bijection $f_{F}$ as needed which ends the proof.
		\end{itemize}}
	\end{proof}
	
	\par{The contrast condition is even clearer than the definition of XVCS itself:
		\begin{lemma}
			Given an access structure $(\Gamma_{Qual},\Gamma_{Forb})$. $S_{0}$ and $S_{1}$ satisfies contrast condition $\iff \omega(B_{0i}[\{s\}])<\omega(B_{1j}[\{s\}])(\forall i,j=1,2,\cdots,k; \forall s=1,2,\cdots,t)$. 
		\end{lemma}
	}
	
	\par{As we mentioned in Section \ref{s3}, the following corollary is of great importance in this section. }
	\begin{corollary}[XVCS Insertion Theorem]\label{insertion_theorem_XVCS}
		\par{All XVCS can be inserted into an XVCS constructed from \eqref{w_eqs} and \eqref{b_eqs}. That is, given an XVCS, we can construct another XVCS from \eqref{w_eqs} and \eqref{b_eqs} without influence on the reconstructed picture. }
	\end{corollary}
	\begin{proof}
		\par{Given an XVCS with $2$ collections of basis matrices $C_{0}$ and $C_{1}$. By \ref{better_security_condition}, suppose that $|C_{0}|=|C_{1}|=k$ and $C_{0}=\{X_{01},X_{02},\cdots,X_{0k}\},C_{1}=\{X_{11},X_{12},\cdots,X_{1k}\}$. Define $B_{ij}:=\mathbf{\Gamma_{Qual}}X_{ij}(i=0,1;j=1,2,\cdots,k)$. Now consider these $2k$ linear systems: $\mathbf{\Gamma_{Qual}}X=B_{ij}(i=0,1;j=1,2,\cdots,k)$. Clearly, the condition in lemma \ref{security_theorem_complete} is satisfied. Also, the contrast condition is satisfied. So the solution sets $S_{0}$ and $S_{1}$ constructs an XVCS, as needed.}
	\end{proof}
	
	\begin{remark}
		\par{This corollary claims that all research on XVCS can be translated into research on XVCS constructed from \eqref{w_eqs} and \eqref{b_eqs}. As is known to all, the structure of linear systems' solution set is much clearer than several sophisticated matrices. From now on, any XVCS mentioned in this paper will be seen as an XVCS constructed from \eqref{w_eqs} and \eqref{b_eqs}. }
	\end{remark}
	
	\subsection{SXVCS and SemiSXVCS}
	\subsubsection{Main Theorem}
	\par{It is an immediate idea that if the given XVCS is a SemiSXVCS, then for any $F\in \Gamma_{Forb}$, which $f_{F}$ to choose is actually inessential since all $S_{0j}$(or all $S_{1j}$) looks completely same. With this idea, we can easily get the theorem followed. }
	\begin{theorem}[SXVCS and SemiSXVCS]\label{SXVCS_and_SemiSXVCS}
		\par{All SemiSXVCS can be viewed as a composition of several SXVCS in the view of linear systems. That is, if a SemiSXVCS is constructed from \eqref{w_eqs} and \eqref{b_eqs} witch contains $2k$ linear systems, then we can construct $k$ SXVCS from these linear systems.}
	\end{theorem}
	\begin{proof}
		\par{Given a White-SemiSXVCS, by corollary \ref{insertion_theorem_XVCS}, it can be viewed as an XVCS constructed from \eqref{w_eqs} and \eqref{b_eqs}. We randomly choose $2$ solution sets $S_{0i}$ and $S_{1j}$. Then the $2$ sets satisfy the contrast condition. By lemma \ref{security_theorem_complete}, $\forall F\in\Gamma_{Forb}$, $S_{0i}[F]=S_{1f_{F}(i)}[F]$. Notice that $S_{0i}=S_{0j}(\forall i,j)$, it is easy to find that $S_{1j}[F]=S_{01}[F](\forall j=1,2,\cdots,k)$. So $S_{0i}$(or just write $S_{01}$) and $S_{1j}$ constructs an SXVCS for all $j=1,2,\cdots,k$. This decomposes a SemiSXVCS into several SXVCSs.}
	\end{proof}
	
	\begin{remark}
		\par{In fact, not only SemiSXVCS can be decomposed, but all XVCS can be decomposed. However, general XVCS cannot be decomposed into SXVCS unless it is a SemiSXVCS. They can be decomposed into XVCS with smaller "$k$" i.e. XVCS constructed from less linear systems. We can easily obtain that these kind of decomposition is related to the connectivity of bipartite graph witch is related to the rigidity braced rectangular framework.}
	\end{remark}
	
	\par{This is the main theorem of our work. The next are important corollaries of it. Thanks to corollary \ref{insertion_theorem_XVCS}, we can extend our theorem to the fullest.}
	
	\begin{corollary}\label{OPTIMAL}
		\par{If there is a SemiSXVCS, there is an SXVCS on it with the same pixel expansion and the average contrast no smaller than it. That is, in most cases, both the smallest pixel expansion and the greatest contrast can be reached by some SXVCS when it exists.}
	\end{corollary}
	\begin{proof}
		\par{Decomposing the SemiSXVCS into several SXVCSs by \ref{SXVCS_and_SemiSXVCS} and choosing the one with the greatest average contrast will finish the proof.}
	\end{proof}
	
	\begin{remark}
		\par{Attention! It is not right to say the smallest pixel expansion and the greatest average contrast can be reached simultaneously in one SXVCS up to now since we have not proved that yet. The only thing we know is that each of them can be reached by some SXVCS. However, we have already proved that $(2,n)$-XVCS can reach its smallest pixel expansion and greatest contrast in just one SXVCS. We will see that later.}
	\end{remark}

	\subsubsection{SXVCS with Perfect White Pixel Reconstruction}
	
	\par{Notice that PW-XVCS is SemiSXVCS, we may wonder about if there is any result about PW-XVCS.}
	
	\begin{corollary}
		\par{Given a SemiSXVCS with perfect white pixel reconstruction, we can find an SXVCS with perfect white pixel reconstruction and they have the same pixel expansion.}
	\end{corollary}
	\begin{proof}
		\par{Just decompose the SemiSXVCS into several SXVCSs by \ref{SXVCS_and_SemiSXVCS} and choose one.}
	\end{proof}
	
	\begin{remark}
		\par{If PW-XVCS exists, then SemiXVCS exists (since PW-XVCS is SemiSXVCS). If SemiSXVCS exists, then SXVCS exists (by \ref{SXVCS_and_SemiSXVCS}). If SXVCS exists, then PW-SXVCS exists. Finally, if PW-SXVCS exists, then PW-XVCS exists. The four are equivalent in existence.}
	\end{remark}
	
	\par{ In \cite{Tuyls2005}, Tuyls has already constructed $(2,n)$-SemiSXVCS with the smallest pixel expansion(although he had no concept of SemiSXVCS). This shows that perfect white pixel reconstruction is not that rare in previous work. We will see the optimal $(2,n)$-SXVCS later in the paper.}

    \begin{example}[Quick Determination of Perfect White Pixel Reconstruction]\label{Quick Determination of Perfect White Pixel Reconstruction}
        \par{In fact, this is a skillful method so we will not write it in the form of a theorem. Here we take $(3,4)$-XVCS as an example.}
	    \par{By theorem \ref{SXVCS_and_SemiSXVCS}, we have known that SemiSXVCS exists is equivalent to SXVCS exists. So to determine the existence of SemiSXVCS is just to determine the existence of SemiSXVCS. }
     
        \par{Consider the access structure of $(3,4)$-XVCS. The qualified matrix is given by $\Gamma_{Qual}=\left[
				\begin{array}{cccc}
					1&	1&	1&	0\\
					1&	1&	0&	1\\
					1&	0&	1&	1\\
					0&	1&	1&	1\\
				\end{array}
				\right]$. Consider the systems below: 
				\begin{eqnarray*}
					\mathbf{\Gamma_{Qual}}X=\mathbf{0_{4,m}};\\
					\mathbf{\Gamma_{Qual}}X=\left[
					\begin{array}{c}
						p_{1}\\
						p_{2}\\
						p_{3}\\
						p_{4}\\
					\end{array}
					\right].
				\end{eqnarray*}
				where row vector $p_{i}$ is of size $1\times m(i=1,2,3,4)$. Now just suppose that the second equation is consistent which is a necessary condition of the two systems constructing an SXVCS. By security condition, we claim that $p_{1}+p_{2}=\mathbf{0_{1,m}}$ and $p_{1}+p_{2}+p_{3}=\mathbf{0_{1,m}}$. That means $p_{3}=\mathbf{0_{1,m}}$. Similarly, $p_{i}=0(i=1,2,3,4)$. That contradicts the contrast condition. }	
    \end{example}
    \par{This method can be applied to other conditions. For example, we can easily check that all $(3,n)$-SemiSXVCS does not exists($\forall n> 3$) and as a result $(k,n)$-SemiSXVCS does not exists($\forall k\geq 3,n>k$). The result is an extension of \cite{Tuyls2005}'s proposition 6 which claims that $(k,n)$-XVCS with perfect white pixel reconstruction does not exists($\forall k\geq 3,n>k$). Now we have known the deeper reason for this phenomenon: the existence of SXVCS.}
 
	\subsubsection{SXVCS with Pixel expansion $1$}
	\par{The next is the equivalent condition of the existence of XVCS with pixel expansion $1$ \cite{Shen2017}.}
	
	\begin{lemma}\label{perfect_XVCS_old}
		\par{Given an access structure $(\Gamma_{Qual},\Gamma_{Forb})$, there exists an XVCS with pixel expansion $1$ if and only if the access structure satisfies the $2$ conditions:
			\begin{itemize}
				\item The result of the stack of odd numbers of rows in $\mathbf{\Gamma_{Qual}}$ is not a row in $\mathbf{\Gamma_{Forb}}$;
				\item If the result of the stack of even numbers of rows in $\mathbf{\Gamma_{Qual}}$ corresponds to the bitset $P_{0}$, then $P_{0}=\emptyset$ or $P_{0}\in\{V|\nexists Q\in \Gamma^{-}, V\subseteq Q\}.$ 
		\end{itemize}}
		
	\end{lemma}
	
	\par{Here, we have to point out that this is not the simplist form and we will improve the lemma.}
	
	\begin{lemma}\label{perfect_XVCS_new}
		\par{Given an access structure $(\Gamma_{Qual},\Gamma_{Forb})$, there exists an XVCS with pixel expansion $1$ $\iff$ the result of the stack of odd numbers of rows in $\mathbf{\Gamma_{Qual}}$ is not a row in $\mathbf{\Gamma_{Forb}}$;}
	\end{lemma}
	\begin{proof}
		\par{Suppose that the second point in lemma \ref{perfect_XVCS_old} is not true. Then there exists even numbers of rows of $\mathbf{\Gamma_{Qual}}$ of which the corresponding bitset of the result of the stack(denoted by $A$) is a none empty subset of $Q\in \Gamma^{-}$. Then define $\mathbf{B}:=\mathbf{A}+\mathbf{Q}$. So $B$ is a proper subset of $Q$ since $\emptyset\neq A\subseteq Q$. However, since $\mathbf{A}$ is the stack result of even numbers of rows in $\mathbf{\Gamma_{Qual}}$, $\mathbf{B}$ should be the stack result of odd numbers of rows in $\mathbf{\Gamma_{Qual}}$. That means $B\notin \Gamma_{Forb}$. This contradicts with $B\subsetneqq Q\in \Gamma^{-}$.}
	\end{proof}

	\section{Application: The Optimal $(2,n)$-XVCS}\label{s5}
	\subsection{Three optimal Parameters}
	\par{Unlike $(k,k)$-XVCS or other XVCS with pixel expansion $1$, $(2,n)$-XVCS cannot just be described as absolutely "perfect" or absolutely "optimal". The "optimal" here is defined as the optimal pixel expansion, average contrast, and noise of the reconstructed picture among which the pixel expansion and the average contrast are the most important parameters to evaluate an XVCS and the noise of the reconstructed picture is mentioned as we happened to discover it. However, if you really do not care about these three important parameters, you can define other parameters to evaluate an XVCS(although that is hard to imagine). But even if you did that, this section still helps a lot.}
	
	\par{In previous research, the researchers may find the optimal pixel expansion of $(2,n)$-XVCS or the optimal average contrast of it but they cannot reach both of them at the same time, not to mention ensure that it is an SXVCS! }
	
	\par{Thanks to the preliminary study of SXVCS and SemiSXVCS, we can transform the study of general $(2,n)$-XVCS into $(2,n)$-SXVCS.}
	
	\subsection{The Algorithm}
	\par{Here we will show the algorithm for constructing an optimal $(2,n)$-XVCS. It is a simple algorithm with low complexity(no more than $O(n)$). We will show the proof of it later.}
	\par{
		\begin{algorithm}[htb]
			\caption{The Algorithm of Constructing an Optimal $(2,n)$-XVCS.}
			\label{optimal_(2,n)_XVCS}
			\begin{algorithmic}[1]
				\Require
				The number $n$ in $(2,n)$-XVCS;
				\Ensure The basis matrices of optimal $(2,n)$-XVCS;
				\State Initially matrix $B_{1}^{*}=[1]$, $m=1$, $i=2$;
				\State Let $i=i+1$;\label{Step_2}
				\State If $i=2^{k}+1$ for some $k=0,1,2,3,\cdots$, let $B_{1}^{*}=\left[\begin{array}{cc}
					B_{1}^{*}&	\mathbf{1_{i-1,1}}\\
					\mathbf{0_{1,m}}&	1
				\end{array}\right]$ and $m=m+1$;
				
				\State Else, $B_{1}^{*}=\left[\begin{array}{c}
					B_{1}^{*}\\
					B_{1}^{*}[\{2^{m}-i+1\}]
				\end{array}\right]$;
				
				\State If $i\neq n$, go back to step 2;
				
				\State Solve the $2$ linear systems 
				$$\left[\begin{array}{ccccc}
					1&1&&&\\
					&1&1&&\\
					&&\ddots&\ddots&\\
					&&&1&1\\
				\end{array}\right]_{n-1,n}X=\mathbf{0_{n-1,m}};$$
				$$\left[\begin{array}{ccccc}
					1&1&&&\\
					&1&1&&\\
					&&\ddots&\ddots&\\
					&&&1&1\\
				\end{array}\right]_{n-1,n}X=B_{1}^{*};$$
				
				\Return The $2$ solution sets;
			\end{algorithmic}
		\end{algorithm}
	}	
	\par{In fact, not only can we construct an optimal $(2,n)$-XVCS in $O(n)$ but we can also construct $n-1$ optimal $(2,k)$-XVCSs for all $k=2,3,\cdots,n$ in at most $O(n^{2})$ since $B_{1}^{*}$ in the algorithm is unique for a fixed $n$ and we have already gotten all $B_{1}^{*}$ for $k=2,3,\cdots,n$ on the way of  constructing $B_{1}^{*}$ for $n$ and the only thing left for us to do is to solve the $2$ linear systems determined by different $B_{1}^{*}$ which is expected to be finished in $O(k)$ since the coefficient matrix is written in a form that is easy to solve. }
	
	\par{To prove the algorithm, we have to introduce $(2,n)$-PXVCS first, then calculate the maximum contrast of it which solves the main problem in $(2,n)$-PXVCS, prove that PXVCS and SXVCS have the same maximum contrast and at last calculate the maximum contrast of $(2,n)$-SXVCS and finish the proof.}

	\subsection{Optimal $(2,n)$-PXVCS}
	\par{Probability-based XVCS(PXVCS) is mentioned here. Although $(2,n)$-PXVCS is not a core concept of this paper, it helps to simplify the proof. }
		
		\begin{definition}\label{PXVCS}
			\par{Given an access structure $\Gamma=(\Gamma_{Qual},\Gamma_{Forb})$ with the same notations above. If two (finite) collections of $n\times m$ Boolean matrices $C_{0}$,$C_{1}$ satisfy the following two conditions:
				\begin{center}
					\begin{itemize}
						\item \textbf{Contrast Condition:} $\forall Q\in \Gamma_{Qual}$, $\displaystyle{\dfrac{1}{|C_{0}|}\sum_{M_{0}\in C_{0}}\omega(\oplus M_{0}[Q])<\dfrac{1}{|C_{1}|}\sum_{M_{1}\in C_{1}}\omega(\oplus M_{1}[Q])}$
						
						\item \textbf{Security Condition:} $\forall F=\{i_{1},i_{2},\cdots,i_{p}\}\in \Gamma_{Forb}$, the two collections of $p\times n$ matrices $D_{0}$ and $D_{1}$ obtained by restricting $C_{0}$ and $C_{1}$ to rows $i_{1},i_{2},\cdots,i_{p}$, respectively, denoted by $D_{0}=C_{0}[F]$ and $D_{1}=C_{1}[F]$, are distinguishable in the sense that they contain the same matrices with the same frequencies.
					\end{itemize}
				\end{center}
				then we say $C_{0}$ and $C_{1}$ construct a PXVCS on $\Gamma$. 
			}
		\end{definition}
		
		\par{Similar to general XVCS, we can insert PXVCS into the solution sets of $2k$ linear systems.}
		
		\begin{lemma}
			All PXVCS can be inserted into a PXVCS constructed from \eqref{w_eqs} and \eqref{b_eqs}.
		\end{lemma}
		
		\par{Also similarly, the $2k$ systems construct a PXVCS if and only if it fits the security condition in lemma \ref{security_theorem_complete} and fits the contrast condition $\displaystyle{\sum_{i=1}^{k}\sum_{s=1}^{t}\omega(B_{0i}[\{s\}])<\sum_{i=1}^{k}\sum_{s=1}^{t}\omega(B_{1i}[\{s\}])}$. Since it is almost same with the proof given before, we just skip it.}
		
		\par{However, we will show some new thoughts here.}
		
		\subsubsection{The Optimal Contrast for $(2,n)$-PXVCS with Pixel Expansion $1$}
		\par{If we restrict the problem to pixel expansion $1$ and with the perfect white pixel reconstruction, the problem will be much easier to solve by probability theory.}
		
		\par{Now suppose that the probability of choosing $\mathbf{0_{n,1}}\in C_{0}$ while we are applying a $(2,n)$-PXVCS is $p$ while the probability of choosing $\mathbf{1_{n,1}}\in C_{0}$ is $1-p$. For matrices in $C_{0}$, we suppose that it looks like:
			$$\left[\begin{array}{c}
				\xi_{1}\\
				\xi_{2}\\
				\vdots\\
				\xi_{n}
			\end{array}\right],$$ where $\xi_{i}$ obeys Bernoulli distribution with $P(\xi_{i}=0)=p$ and $P(\xi_{i}=1)=1-p$. Take attention that they can be not independent of each other(in fact, if independent, the maximal contrast will be $\frac{1}{2}$ \cite{2007Two}). The average contrast is defined by $\alpha=\dfrac{\displaystyle{\sum_{1\leq i<j\leq n}E(\xi_{i}\oplus \xi_{j})}}{C_{n}^{2}}$.  We denote that $P_{r_{1}r_{2} \cdots r_{n}}:=P(\xi_{1}=r_{1},\xi_{2}=r_{2}, \cdots ,\xi_{n}=r_{n}),r_{i}\in\{0,1\}$.
		}
		
		\par{Then we calculate the upper boundary of $\alpha$.}
		\begin{eqnarray*}
			\alpha&&=\dfrac{\displaystyle{\sum_{1\leq i<j\leq n}E(\xi_{i}\oplus \xi_{j})}}{C_{n}^{2}}\label{1}\\
			&&=\dfrac{
				\displaystyle{
					\sum_{1\leq i<j\leq n}
					\sum_{r_{i}\oplus r_{j}=1}P_{r_{1}r_{2} \cdots r_{n}}
				}
			}
			{
				C_{n}^{2}
			}\label{2}\\
			&&=\dfrac{
				\displaystyle{
					\sum_{r_{1}r_{2}..r_{n}}\omega(r_{1}r_{2}..r_{n})[n-\omega(r_{1}r_{2}..r_{n})]P_{r_{1}r_{2} \cdots r_{n}}
				}
			}
			{
				C_{n}^{2}
			}\label{3}\\
			&&\leq \dfrac{
				\displaystyle{
					[\frac{n}{2}][\frac{n+1}{2}]\sum_{r_{1}r_{2}..r_{n}}P_{r_{1}r_{2} \cdots r_{n}}
				}
			}
			{
				C_{n}^{2}
			}\label{4}\\
			&&=\dfrac{
				[\frac{n}{2}][\frac{n+1}{2}]
			}
			{
				C_{n}^{2}
			}\label{5}
		\end{eqnarray*}
		\par{The binary code $r_{1}r_{2}\cdots r_{n}$ is seemed as a binary vector in $\omega(r_{1}r_{2}\cdots r_{n})$. It is easy to check that the $"="$ can be reached in the above proof.}
		
		\par{No wonder the maximal contrast of $(2,n)$-PXVCS is just $\dfrac{[\frac{n}{2}][\frac{n+1}{2}]}{C_{n}^{2}}$.}
		
		\subsubsection{PXVCS and SXVCS}
		\par{Here we just have to know that the optimal average contrast of  $(2,n)$-SXVCS is no larger than $(2,n)$-PXVCS. In fact, they are equal.}
		
		\begin{lemma}
			Given SXVCS, there exists PXVCS with average contrast equal to the given SXVCS.
		\end{lemma}
		\begin{proof}
			\par{Just consider the systems \eqref{white_SXVCS} and \eqref{black_SXVCS}. Suppose that $B_{i}=[B_{i1},B_{i2},\cdots.B_{im}](i=0,1)$ where $m$ is the pixel expansion of the given SXVCS. Substitute these constant matrices into \eqref{w_eqs} and \eqref{b_eqs}. It is easy to check that the security condition and contrast condition of PXVCS are satisfied. Also, the average contrast of this PXVCS is equal to the given SXVCS.}
		\end{proof}
		\par{As a result, we have gotten the upper boundary of $(2,n)$-SXVCS and by corollary \ref{OPTIMAL} and \ref{general_security_theorem_mat_SXVCS}, we can easily conclude that it is also the upper boundary of $(2,n)$-XVCS! }
		
		\subsection{The Proof}
		\par{It is already known that the minimal pixel expansion of $(2,n)$-XVCS is $\lceil\log_{2}n\rceil$ \cite{Tuyls2005}. It is also easy to check that in SXVCS but we skip that for convenience. Here just check that the scheme we gave before has the optimal pixel expansion.}
		
		\par{The noise of our scheme is no wonder zero as well.}
		\par{Thus, the only thing to check is that the average contrast of our construction is optimal(i.e. $\dfrac{[\frac{n}{2}][\frac{n+1}{2}]}{C_{n}^{2}}$).}
		
		\begin{proof}
			\par{\begin{itemize}
					\item The contrast condition of all elements of qualified set of $(2,n)$ can be described by the stack result of continuous rows of $B_{1}^{*}$ we get(i.e. $[0,0,\cdots,0,\stackrel{l\text{th}}{1},0,\cdots,0,\stackrel{h\text{th}}{1},0,\cdots,0]X=\oplus(B_{1}^{*}[\{l,l+1,\cdots,h\}])$). Notice that $B_{1}^{*}=\left[\begin{array}{c}
						B_{1}^{*}\\
						B_{1}^{*}[\{2^{m}-i+1\}]
					\end{array}\right]$ is in fact adding a row by symmetry of $B_{1}^{*}[\{2^{m}-i+1\}]$ by $B_{1}^{*}[\{2^{m-1}\}]$, it is easy to check that each $\oplus(B_{1}^{*}[\{l,l+1,\cdots,h\}])$ is not zero by mathematical induction. This indicates that the construction is an SXVCS.
					\item We just restrict the proof to a column of $B_{0}^{*}$, for example, the first column, denoted by $B_{01}^{*}$. We just have to check that $\displaystyle{\sum_{i=1}^{n-1}\sum_{j=i}^{n-1}\oplus(B_{01}^{*}[\{i,i+1,\cdots,j\}])=[\frac{n}{2}][\frac{n+1}{2}]}$. By mathematical induction, it is easy to get the result.
			\end{itemize}}
		\end{proof}
		\par{Thus we have already proved that our construction gives an optimal $(2,n)$-XVCS(i.e. a $(2,n)$-SXVCS with pixel expansion $\lceil\log_{2}n\rceil$ and average contrast $\dfrac{[\frac{n}{2}][\frac{n+1}{2}]}{C_{n}^{2}}$).}
		
\section{Experiments and Comparison}
In this section, we presents some experiments and comparisons related to SXVCS.

\subsection{Some Access Structures with SXVCS}
Table \ref{Feasibility in Constructing SXVCS} displays the existence of SXVCS on several types of Access Structures. 

Here, we say an access structure satisfying the conditions of Theorem \ref{perfect_XVCS_new} is a perfect structure\label{perfect structure}. 
Also, given the access structure $\Gamma$ with qualified matrix $\mathbf{\Gamma_{Qual}}$, for all $i=1,2,\cdots,m$, replace the $1$ in the $i$-th column of the matrix with a row vector of length $m_{i}$ consisting of all $1$s, and replace the $0$ with a row vector of length $m_{i}$ consisting of all $0$s. This results in a new matrix, which, as the access structure corresponding to the qualified matrix, is called derived from $\Gamma$\label{derive}.

It is necessary to underline that most access structures cannot construct SXVCS on them. The table below just shows that SXVCS is not that rare instead of showing that it is common to see access structures with SXVCS.

\begin{table}[h] 
\centering
\caption{Access Structures and Their Feasibility in Constructing SXVCS.} 
\label{Feasibility in Constructing SXVCS}
\begin{tabular}{cc}
\toprule
  Access structure & Can Construct SXVCS \\
\hline
$(2,n)(n\geq 2)$ & Yes   \\ 
$(n,n)$ & Yes  \\ 
$(k,n)(n>k>2)$ & No \\
perfect structure\ref{perfect structure} & Yes \\
access structures derived from other access structures with SXVCS\ref{derive} & Yes\\ \hline
\end{tabular}
\end{table}


\subsection{Experiments}
In this section, we show the secret image, the shares and the reconstructed images of SXVCS and SemiSXVCS. We take $(2,3)$-SXVCS and PW-$(2,3)$-XVCS as an example.

Let's begin by examining a $(2,3)$-SXVCS instance. The secret image, the shares and the reconstructed images are presented in Figure \ref{fig:example}. It is easy to check that there is no noise in all the reconstructed images. In this instance, it appears that the black regions of the Secret Image (SI) are depicted by vertical black and white lines in the reconstructed images, while the white pixels remain white.

	\begin{figure}[h]
		\centering
        \subfloat[Secret Image]{\includegraphics[width=1in]{"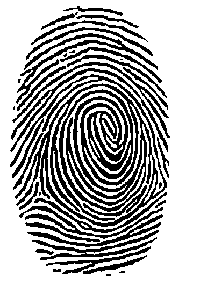"}\label{fig:si}} \hspace{0.03in}
		\subfloat[SXVCS1]{\includegraphics[width=2in]{"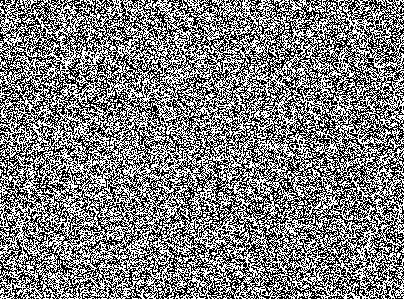"}\label{fig:sx1}}\hspace{0.03in}
		\subfloat[SXVCS2]{\includegraphics[width=2in]{"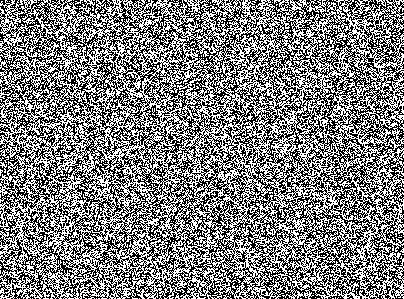"}\label{fig:sx2}}\\
		\subfloat[SXVCS3]{\includegraphics[width=2in]{"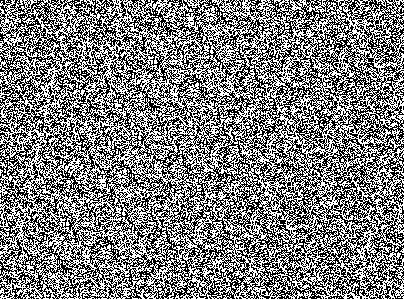"}\label{fig:sx3}}\hspace{0.03in}
		\subfloat[SXVCS2+3]{\includegraphics[width=2in]{"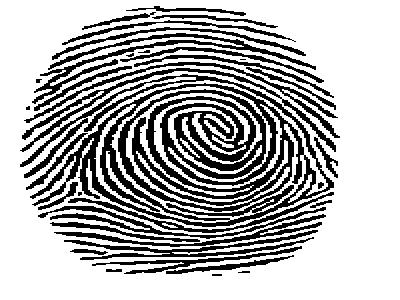"}\label{fig:sx23}} \hspace{0.03in}
		\subfloat[SXVCS1+2]{\includegraphics[width=2in]{"figures//SXVCS_1+2.png"}\label{fig:sx12}} \hspace{0.03in}
		\subfloat[SXVCS1+3]{\includegraphics[width=2in]{"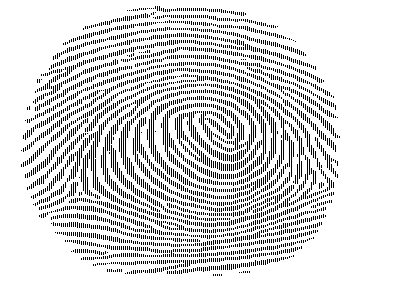"}\label{fig:sx13}}
		\caption{An experiment of $(2,3)$-SXVCS}
		\label{fig:example}
	\end{figure}
	

Next, let's examine a PW-$(2,n)$-XVCS that does not conform to the SXVCS structure in Figure \ref{fig:PW-XVCS1 example}. We omit the Secret Image (SI) and shares in this instance as they bear similarities to those in the (2,3)-SXVCS. It's evident that the reconstructed image is permeated with noise, even though the Secret Image (SI) can still be discerned from it. Consequently, PW-SXVCS provides a clearer representation of the details compared to a PW-XVCS that does not embody the SXVCS structure.

\begin{figure}[h]
		\centering
        \subfloat[PW-XVCS1+2]{\includegraphics[width=2in]{"figures//SemiSXVCS_1.png"}\label{fig:si}} \hspace{0.03in}\hspace{0.03in}
		\subfloat[PW-XVCS1+3]{\includegraphics[width=2in]{"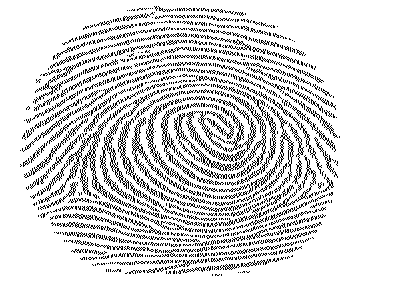"}\label{fig:sx1}}\hspace{0.03in}\hspace{0.03in}
		\subfloat[PW-XVCS2+3]{\includegraphics[width=2in]{"figures//SXVCS_2+3.png"}\label{fig:sx2}}\hspace{0.03in}
		\caption{An experiment of $(2,3)$-PW-XVCS(SemiSXVCS)}
		\label{fig:PW-XVCS1 example}
	\end{figure}

\subsection{Comparisons for $(2,n)$-XVCS}
 Here we first list Table \ref{Several $(2,n)$-XVCS} to compare our $(2,n)$-XVCS and $(2,n)$-PXVCS with pixel expansion $1$ to other researchers' work.

 The "optimal $(2,n)$-XVCS with random row permutation" in the table is constructed from the optimal $(2,n)$-XVCS we mentioned before. Given the basic matrices $S_{0}$ and $S_{1}$ of the optimal $(2,n)$-XVCS we mentioned before, then let $C_{i}$ collects the random row permutation of matrices in $S_{i}(i=0,1)$. It is easy to check that $C_{0}$ and $C_{1}$ construct a new XVCS with parameters listed in the table.

\begin{table}[h]
\centering
\caption{The performance comparison of several $(2,n)$-XVCS} 
\label{Several $(2,n)$-XVCS}
\begin{tabular}{cccccc}
\toprule
  Schemes & Pixel & Average & Minimal  & Noise  & Is PXVCS\\
    & Expansion & Contrast& Contrast &   & \\
\hline
Our optimal & $\lceil\log_{2}n\rceil$ & $\dfrac{[\frac{n}{2}][\frac{n+1}{2}]}{C_{n}^{2}}$ & $\frac{1}{\lceil\log_{2}n\rceil}$ & No & No\\ 


Our $(2,n)$-PXVCS & $1$ & $\dfrac{[\frac{n}{2}][\frac{n+1}{2}]}{C_{n}^{2}}$& -  & Large & Yes \\

Our optimal (random row) &$\lceil\log_{2}n\rceil$ & $\dfrac{[\frac{n}{2}][\frac{n+1}{2}]}{C_{n}^{2}}$ & $\dfrac{[\frac{n}{2}][\frac{n+1}{2}]}{C_{n}^{2}}$  & Large & No\\ 

Tuyls\cite{Tuyls2005} & $\lceil\log_{2}n\rceil$ & $\geq\frac{1}{\lceil\log_{2}n\rceil}$* & $\geq\frac{1}{\lceil\log_{2}n\rceil}$  & Large & No \\

Fu-Yu\cite{fu2014optimal} & $\lceil\log_{2}n\rceil$ & $\geq\frac{1}{\lceil\log_{2}n\rceil}$**&$\geq\frac{1}{\lceil\log_{2}n\rceil}$  & Large & No \\

Daoshun Wang et al.\cite{2007Two} & $1$  & $\frac{1}{2}$& -  & Large & Yes\\

\hline


\end{tabular}
\begin{minipage}{\textwidth}
\vspace{0.5em}
\footnotesize{
*: Here we just show his $(2,n)$-XVCS with minimun pixel expansion and other schemes are ignored. The definition of contrast in \cite{Tuyls2005} can be seen aes $\displaystyle\min_{Q\in \Gamma_{Qual}}\{\alpha(Q)\}$ in our paper. 

**: Thier scheme relies on the other schemes, and the average contrast has not calculated yet.
}
\end{minipage}
\end{table}

\begin{figure}[h]
		\centering
        \subfloat[Secret Image]{\includegraphics[width=0.5in]{"figures//SI.png"}\label{fig:si}} \hspace{0.03in}
		\subfloat[SXVCS1]{\includegraphics[width=1in]{"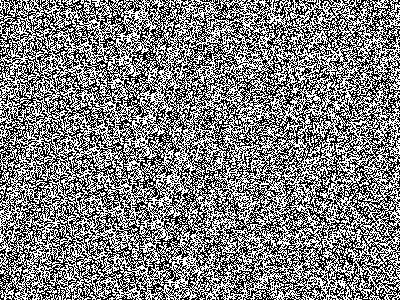"}\label{fig:sx1}}\hspace{0.03in}
		\subfloat[SXVCS2]{\includegraphics[width=1in]{"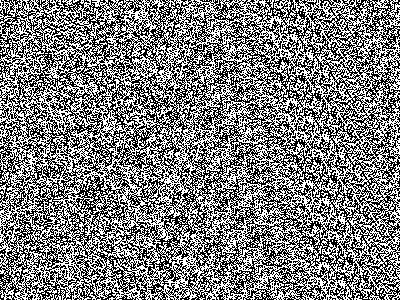"}\label{fig:sx2}}\hspace{0.03in}
		\subfloat[SXVCS3]{\includegraphics[width=1in]{"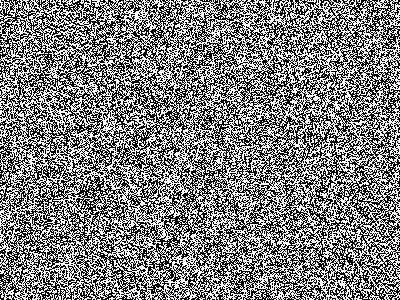"}\label{fig:sx3}}\\
		\subfloat[SXVCS4]{\includegraphics[width=1in]{"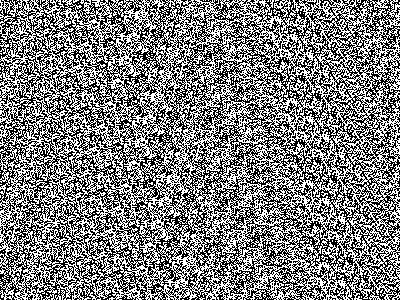"}\label{fig:sx23}} \hspace{0.03in}
		\subfloat[SXVCS1+2]{\includegraphics[width=1in]{"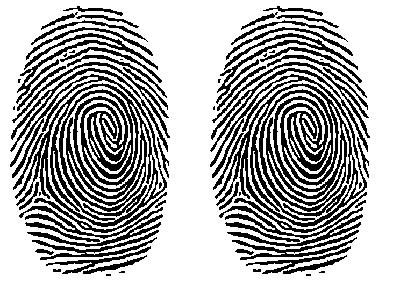"}\label{fig:sx12}} \hspace{0.03in}
		\subfloat[SXVCS1+3]{\includegraphics[width=1in]{"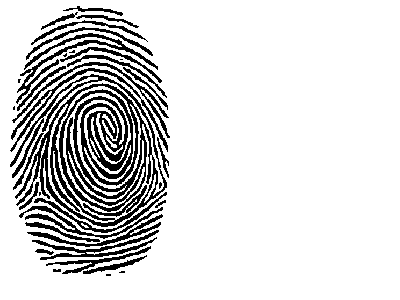"}\label{fig:sx13}}\hspace{0.03in}
            \subfloat[SXVCS1+4]{\includegraphics[width=1in]{"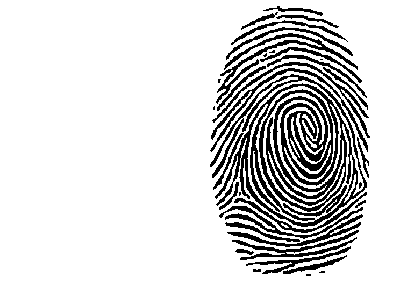"}\label{fig:sx3}}\\
            \subfloat[SXVCS2+3]{\includegraphics[width=1in]{"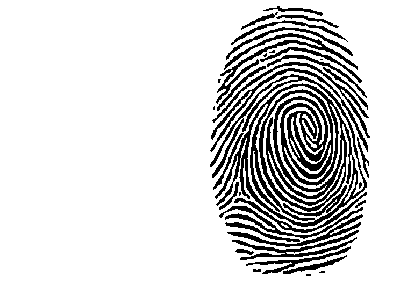"}\label{fig:sx3}}\hspace{0.03in}
            \subfloat[SXVCS2+4]{\includegraphics[width=1in]{"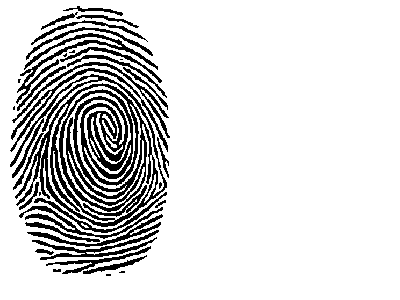"}\label{fig:sx3}}\hspace{0.03in}
            \subfloat[SXVCS3+4]{\includegraphics[width=1in]{"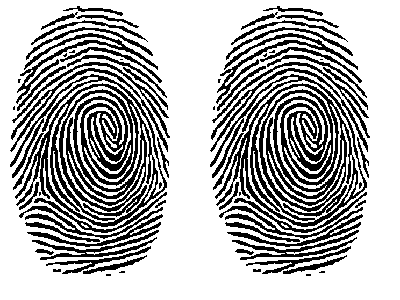"}\label{fig:sx3}}\\
            \subfloat[$(2,n)$-XVCS region by region scheme in \cite{Shen2017} ]{\includegraphics[width=4in]{"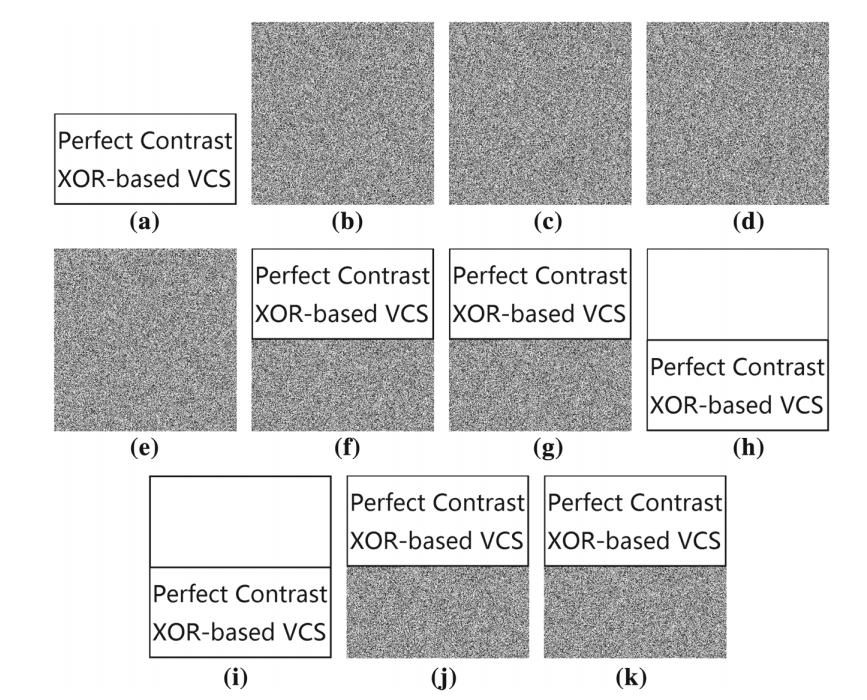"}\label{fig:sx3}}
		\caption{Comparison between our $(2,n)$-SXVCS and \cite{Shen2017}'s XVCS region by region}
		\label{SemiSXVCS and SXVCS region}
	\end{figure}

As to the scheme in \cite{Shen2017}, their region-by-region scheme is in fact multiple but each participant can stuck the shares in hand into one share with several regions.
It is easy to check that given an SXVCS we can construct an XVCS region-by-region immediately by just distribute each subpixel to each region. The only difference is that for all $Q\in\Gamma_{Qual}$, the recovered image of us has no noise region but just white region and SI region while the recovered image of \cite{Shen2017} has noise region.
Here we compare our $(2,4)$-XVCS with \cite{Shen2017} in Figure \ref{SemiSXVCS and SXVCS region}. Imagine that the finger point is replaced by noise image, our scheme still works but \cite{Shen2017}'s not. Also, the pixel expansion of our region by region scheme is $\lceil\log_{2}n\rceil$ while \cite{Shen2017}'s is $\lceil n\rceil$



 		\section{Conclusion and Future Work}
            
In this paper, we establish some basic theory of constructing XVCS through linear equations systems (or matrix equations), and prove that any XVCS can be constructed from $2k$ linear equations systems. Specifically, we extract XVCS constructed from only two linear equations systems as XVCS without noise, or SXVCS, and study its basic properties. We obtain the equivalence theorem among SXVCS, SemiSXVCS, PW-SXVCS and PW-XVCS, and thus introduce the idea of optimizing PW-XVCS construction through SXVCS. Inspired by this idea, we obtain the construction method of optimal (2.n)-XVCS through some further simple calculations. The (2,n)-XVCS constructed by this method performs very well in most aspects.

Overall, there have been some good progress in the study of XVCS in noise recovery of images in this paper. We have revealed which access structures have SXVCS and demonstrated that this property can significantly optimize solutions by convert PW-XVCS to SXVCS.


However, there are still intriguing questions to be addressed. For instance, the assertion "The existence of SXVCS is equivalent to the existence of PW-XVCS" effectively transforms the problem related to PW-XVCS into an issue concerning SXVCS. Even though SXVCS is easier to study, we have yet to provide complete characterization of access structures that clarifies which structures can support SXVCS or PW-XVCS. Furthermore, in terms of constructing the optimal $(2,n)$-XVCS, we have laid out all the achievable results in the table, yet some parameters remain challenging to optimize. These aforementioned unresolved problems are exciting directions for future exploration.


 \bibliography{sn-bibliography}

\end{document}